\documentclass[reqno,a4paper]{amsart}
\renewcommand{\paragraph}{\subsubsection*}
\usepackage[UKenglish]{babel}
\usepackage{microtype}
\usepackage{booktabs}

\usepackage[foot]{amsaddr}
\usepackage{lineno}
\usepackage{url}
\usepackage[
  pdftex,
  bookmarks=true,
  bookmarksnumbered=true,
  hypertexnames=false,
  breaklinks=true,
  linkbordercolor={0 0 1}
  ]{hyperref}

\providecommand{\urlstyle}[1]{}
\providecommand{\doi}[1]{\href{http://dx.doi.org/#1}{\nolinkurl{doi:#1}}}

\usepackage{xcolor}
\usepackage{multirow}
\usepackage{tabularx}
\usepackage{tikz}
\usetikzlibrary{shapes.geometric}
\usepackage{enumitem}
\setlist[enumerate]{noitemsep, topsep=0pt}
\newcommand{\mysqrt}[2]{\sqrt[\leftroot{2}\uproot{5}#1]{ #2}}

\setlist[itemize]{noitemsep, topsep=0pt}
\usepackage{amsmath}
\usepackage{amsfonts}
\usepackage{amsthm}
\usepackage{mathtools}

\newcommand{\NN}{\mathbb{N}}
\newcommand{\ZZ}{\mathbb{Z}}
\newcommand{\QQ}{\mathbb{Q}}

\newcommand{\CC}{\mathbb{C}}

\newcommand{\FF}{\mathbb{F}}
\newcommand{\OO}{\mathcal{O}}
\newcommand{\mf}[1]{\mathfrak{#1}}
\DeclareMathOperator{\lcm}{lcm}
\DeclareMathOperator{\Tr}{Tr}
\DeclareMathOperator{\Disc}{Disc}
\DeclareMathOperator{\Gal}{Gal}

\usepackage{cleveref}
\usepackage{caption}
\usepackage{subcaption}
\usepackage{float}
\usepackage{thm-restate}

\newtheorem{theorem}{Theorem}
\newtheorem{example}{Example}
\newtheorem{lemma}{Lemma}

\newtheorem{proposition}[lemma]{Proposition}

\newcommand{\coNP}{\mbox{{\sf coNP}}}

\newcommand{\coRP}{\mbox{{\sf coRP}}}

\newcommand{\PSPACE}{\mbox{{\sf PSPACE}}}

\usepackage{color, colortbl}
\definecolor{Gray}{gray}{0.9}
\definecolor{LightCyan}{rgb}{0.88,1,1}

\newcommand{\bpp}{\mbox{{\sf BPP}}}

\usepackage{amssymb}

\usepackage{pgf}
\usepackage{pgfplots}

\begin{document}

\title{Identity Testing for Radical Expressions}
\author{Nikhil Balaji,$^1$ Klara Nosan,$^2$, \mbox{Mahsa
  Shirmohammadi,$^2$} \and \mbox{James Worrell$^3$}}
\address{\begin{minipage}{\textwidth}
$^1$~IIT Delhi, New Delhi, India\\
$^2$~Universit\'e Paris Cité, CNRS, IRIF, F-75013 Paris, France\\
$^3$~Department of Computer Science, University of Oxford, UK\end{minipage}}
\date{}

%
%
%

\maketitle

\DeclareRobustCommand{\gobblefive}[5]{}
\DeclareRobustCommand{\gobblenine}[9]{}
\newcommand*{\SkipTocEntry}{\addtocontents{toc}{\gobblenine}}

\renewcommand{\shortauthors}{Balaji, Nosan, Shirmohammadi and Worrell}

\begin{abstract}
  We study the \emph{Radical Identity Testing} problem (RIT): Given an algebraic circuit representing a polynomial $f\in \mathbb{Z}[x_1, \dots, x_k]$ and nonnegative integers $a_1, \dots, a_k$ and $d_1, \dots,$ $d_k$, written in binary, test whether  the polynomial vanishes at the \emph{real radicals} $\mysqrt{d_1}{a_1}, \dots,\mysqrt{d_k}{a_k}$, i.e., test whether $f(\mysqrt{d_1}{a_1}, \dots, \mysqrt{d_k}{a_k}) = 0$. We place the problem in  {\coNP} assuming the Generalised Riemann Hypothesis (GRH), improving on the straightforward {\PSPACE} upper bound obtained by reduction to the existential theory of reals. Next we consider a restricted version, called $2$-RIT, where the radicals are square roots of prime numbers, written in binary. It was known since the work of Chen and Kao~\cite{chen-kao} that $2$-RIT is at least as hard as the polynomial identity testing problem, however no better upper bound than {\PSPACE} was known prior to our work. We show that $2$-RIT is in {\coRP} assuming GRH and in {\coNP} unconditionally. 
Our proof relies on theorems from algebraic and analytic number theory, such as the Chebotarev density theorem and quadratic reciprocity.

\end{abstract}


\maketitle

\section{Introduction}
Identity testing is a fundamental algorithmic question with numerous applications.  In \emph{algebraic} identity testing, the task is to determine the zeroness of an expression evaluated in a given ring.  This problem has many different versions, depending on the syntax for giving the expression and the ring in which the evaluation is to be performed.

A basic instance of algebraic identity testing is the Arithmetic Circuit Identity Testing (ACIT) problem, which involves deciding the zeroness of an integer represented as an arithmetic circuit.  The difficulty in this problem is that the integer may have bit-length exponential in the size of the circuit.  However, the problem admits a randomised polynomial-time algorithm: one evaluates the circuit modulo a prime that is randomly chosen in a certain range.  The ACIT problem turns out to be polynomial-time interreducible with the problem of determining zeroness of an arithmetic circuit evaluated in the  ring of multivariate polynomials: the so-called Polynomial Identity Testing (PIT) problem~\cite{Allender06onthe}.  
Over the years, {PIT} has found diverse applications in algorithm design; a few well-known examples are program testing~\cite{demillo-lipton}, detecting perfect matchings~\cite{lovasz}, factoring polynomials~\cite{kaltofen}, pattern matching in compressed texts~\cite{konig-lohrey, berman-pattern}, primality testing~\cite{agrawal-biswas, aks}, equivalence and minimization of weighted automata~\cite{kiefer-minimize-wa, kiefer-marusic-worrell} and linear recurrence sequences~\cite{COW,abmvv}.
Whether PIT admits a deterministic polynomial-time algorithm is one of the central open questions in complexity theory.

The topic of this paper is \emph{radical identity testing}, that is, testing zeroness of an expression in radicals, represented by an algebraic circuit.
This generalizes the ACIT problem: the evaluation of the circuit occurs in the ring of integers of a number field, rather than the ring of integers of the rational numbers.  Historically, expressions in radicals played a central role in solving polynomial equations.  Such expressions also naturally arise in optimization problems on graphs embedded in Euclidean space, such as the Euclidean Minimal Spanning Tree and Traveling Salesperson problems \cite{ggj}.  Another source of radical expressions arises from an approach by Chen and Kao~\cite{chen-kao} to derandomizing PIT.  Their idea was to test the zeroness of a multivariate polynomial by evaluating it on a certain randomly chosen radical expression.  In their method, by construction, the radical expressions are such that the result of the evaluation is zero if and only if the polynomial is identically zero.  The challenge then becomes to determine the zeroness of the resulting expression in radicals, for which Chen and Kao use numerical approximation.  This approach allowed to reduce the number of random bits required for the problem compared with previously known methods, such as those based on the Schwarz-Zippel lemma.

In this paper, we introduce a symbolic approach to the problem of identity testing algebraic integers in the number field generated over $\mathbb{Q}$ by \emph{real radicals}; given an algebraic circuit  representing a multivariate polynomial $f(x_1,\ldots,x_k)\in\ZZ[x_1,\ldots,x_k]$, and  radical inputs $\mysqrt{d_1}{a_1},\ldots,\mysqrt{d_k}{a_k}$ where the radicands $a_i$, and exponents $d_i$ are nonnegative integers, the \emph{Radical Identity Testing} (RIT) problem is to decide whether $f(\mysqrt{d_1}{a_1},\ldots,\mysqrt{d_k}{a_k}) = 0$. The problem is easily seen to admit a {\PSPACE} upper bound by a reduction to the existential theory of reals~\cite{canny}: introduce a new formal variable for every gate of the circuit, add the equations $x_i^{d_i} - a_i = 0$ and $x_i > 0$ for every radical; RIT is now decided by checking if the resulting system of polynomial equations has a solution over the real numbers.  

Our symbolic algorithm places RIT in {\coNP}, assuming the generalised Riemann hypothesis (GRH).  The main idea behind our algorithm is that if 
$f(\mysqrt{d_1}{a_1},\ldots,\mysqrt{d_k}{a_k}) \neq 0$
then there is a polynomial-length polynomial-time checkable witness of this fact --- namely a prime $p$ and $\overline{\alpha}_1,\ldots,\overline{\alpha}_k \in \mathbb{F}_p$, satisfying $\overline{\alpha}_i ^ {d_i} \equiv a_i \pmod{p}$, such that $\overline{f}(\overline{\alpha}_1,\ldots,\overline{\alpha}_k )$ is non-zero, where $\overline{f}$ is the reduction of $f$ modulo~$p$. Crucial to our approach is the fact of \emph{joint transitivity}, which is the observation that the Galois group of the underlying real field acts jointly transitively on the roots of the various equations $x^{d_i} - a_i = 0$. This allows us to use any of the $d_i$ conjugates $\alpha_i$ of 
$\mysqrt{d_i}{a_i}$
over $\mathbb{F}_p$ in our symbolic algorithm to test whether $\overline{f}(\overline{\alpha}_1,\ldots,\overline{\alpha}_k ) = 0$. In fact, joint transitivity alone can be used in conjunction with a result of Koiran~\cite{KOIRAN1996273} to place RIT in the polynomial hierarchy assuming GRH.  Using Chebotarev's density theorem, and choosing a suitable prime, we improve this upper bound to show that RIT can be solved in {\coNP} assuming GRH.  

We next tackle a special case of RIT namely 2-RIT where the inputs to the circuit are square roots of distinct odd\footnote{
The formulation here corrects the published version~\cite{lics-version} by adding the condition that the prime be odd.} primes. The {PIT} algorithm in~\cite{chen-kao} gives a randomised numerical algorithm for the specific case of 2-RIT involving \emph{bounded} algebraic circuits\footnote{An algebraic circuit of size $s$ computing a multivariate polynomial is called bounded if the degree and bit-length of the coefficients of the polynomials computed at every gate of the circuit is bounded by a polynomial
function in~$s$.} at square roots of distinct prime numbers. However their method completely breaks down when the assumption of boundedness is dropped and the best known upper bound prior to our work was the same as that of the general RIT problem using Koiran's algorithm. We show that 2-RIT is in {\coRP} assuming GRH and in {\coNP} unconditionally. Our techniques are fundamentally different from those of ~\cite{chen-kao};  we use quadratic reciprocity and Dirichlet's theorem on the density of primes in arithmetic progressions.

\paragraph{\bf Related Work}
Balaji et al.~\cite{balaji2021cyclotomic} studied the \emph{cyclotomic identity testing}  problem, where the goal is to determine if an algebraic integer in the cyclotomic number field $\mathbb{Q}(\zeta_n)$ computed by a given algebraic circuit is zero. They show that the problem lies in the complexity class {\bpp} assuming the Generalised Riemann Hypothesis (GRH), and unconditionally in {\coNP}.
We further discuss the approach of \cite{balaji2021cyclotomic} in relation to our work in \Cref{sec:conp}.

Bl\"{o}mer~\cite{blomer98} gave a randomised algorithm to test if a bounded algebraic circuit evaluated at low degree radicals evaluates to zero. Central to such identity questions on algebraic numbers is the knowledge of the Galois group of the extension where the numbers live in. Algorithmic complexity of computing with Galois groups is investigated in ~\cite{landau-miller, landau, arvind-kurur}.

As in~\cite{blomer98}, our formulation of RIT does not permit nesting of radicals, that is, expressions such as $\sqrt{6+4\sqrt{2}}$.  Computational problems associated with nested radicals are treated in~\cite{Blomer92,Blomer97}.

\paragraph{\bf The Sum of Square Roots  problem.} Closely related to $2$-RIT is the square root sum problem where the goal is to infer the sign of a given linear combination of square roots. This is a notorious open problem in numerical analysis and computational geometry. It is known to be decidable in the \emph{Counting Hierarchy} \cite{Allender06onthe}, and the question of determining its precise computational complexity remains open since it was explicitly posed by Garey, Graham and Johnson~\cite{ggj} in 1976. Along with the related problem  of determining the sign of an integer computed by a variable-free algebraic circuit, the square root sum problem is frequently used as a tool for proving hardness and obtaining upper bounds in quantitative verification~\cite{esparza.solving, HKbudget, KWparallel, BLW}, algorithmic game theory~\cite{etessami.yannakakis, UW, CIJ}, formal language theory and logic~\cite{HKL, LOW}. We refer the interested reader to~\cite{problem33, kayal-saha} and the references therein for a discussion on the complexity status of the square root sum problem and related geometric questions.

\paragraph{\bf The Elementary Constant problem.} A related but harder identity testing problem that is a fundamental question at the intersection of transcendental number theory and model theory is the \emph{Elementary Constant problem}~\cite{richardson92}: given a complex number built from rationals using addition, multiplication and exponentiation determine if it is zero. Such numbers are called Elementary numbers, and they form an algebraically closed subfield of the complex numbers. While there is a decision procedure assuming Schanuel's conjecture~\cite{richardson97} for the elementary constant problem, the problem is not known to be decidable unconditionally, and no significant complexity lower/upper bounds are known.

\paragraph{\bf The Compressed Word problem.} 
Arguably the earliest and most fundamental result on identity testing questions are word problems on finitely generated groups and semigroups \cite{dehn1912transformation}. There is a large body of work on the \emph{Compressed Word problem} \cite{lohrey2014compressed} which studies the computational complexity of word problems when the word is represented succinctly via a \emph{straight line program}; see \cite{konig-lohrey, balaji2021cyclotomic} for relations between such word problems and arithmetic identity testing.

\section{Background and Overview}
\label{sec:overview}

In this section, we give a high-level overview of our main results and the main techniques used in our  algorithms.  
 We also introduce some of the  definitions and notations along the
way; we refer the reader to~\Cref{appendix:preliminaries} and~\cite{stewart2001algebraic, StewartBook} for more details. 

\paragraph{Algebraic Circuits.}
Let $X=\{x_1,\ldots,x_k\}$ be a set of commutative variables.
An \emph{algebraic circuit}  over~$X$ is a directed acyclic graph  with labelled vertices and edges. Vertices of in-degree zero (leaves) are labelled with variables in $X$ and $-1$; and the remaining vertices have labels in $\{+,\times\}$. Moreover, the incoming edges to $+$-vertices have labels in $\mathbb{Z}$, that is, the $+$-gates compute integer-weighted sums. There is a unique vertex of out-degree zero which determines the output of the circuit, a $k$-variate polynomial, computed in an obvious bottom-up manner.
The \emph{size} of a circuit is the number of its gates; see Figure~\ref{fig:circuit}.
The \emph{degree} of a circuit $C$ is defined inductively as follows: input gates have degree $1$, the degree of an addition gate is the maximum of the degrees of its inputs, the degree of a multiplication gate is the sum of the degrees of its inputs, and the degree of $C$ is the degree of the output gate. Note that the degree of an algebraic circuit is an upper bound on the degree of its underlying polynomial.
Thus the total degree and the bit-length of the coefficients of a polynomial represented by a circuit is at most exponential in the size of the circuit.

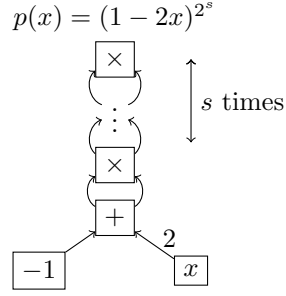
\begin{figure}[t]
\begin{center}	
\begin{tikzpicture} 
\node[draw,label={above:$p(x)=(1-2x)^{2^s}$}] at (2,2.1) (A4) {$\times$};
 \node[draw=none] at (2,1.4) (A3) {$\vdots$};
\draw[<-,  bend left=60]  (A4) edge (A3) ;
\draw[<-,  bend right=60]  (A4) edge  (A3) ; 
  \node[draw] at (2,.7) (A2) {$\times$};
\draw[<-,  bend left=60]  (A3) edge (A2) ;
\draw[<-,  bend right=60]  (A3) edge (A2) ;   
 \node[draw] at (2,0) (A1) {$+$};
\draw[<-,  bend left=60]  (A2) edge (A1) ;
\draw[<-,  bend right=60]  (A2) edge  (A1) ; 

\node[draw] at (1,-.7) (c1) {$-1$};
\node[draw] at (3,-.7) (x) {$x$};

\draw[<-]  (A1) edge (c1) ;
\draw[<-]  (A1) edge node[near start,right]{$~2$} (x) ;
\draw[<->]  (3,2.1) edge node[midway,right]{$s$ times} (3,1) ;
 
\end{tikzpicture}
	\end{center}
	\caption{An  algebraic circuit 
	 computing   the polynomial $p(x)=(1-2x)^{2^s}$, with the highest degree monomial $x^{2^s}$ and the coefficients double exponential $2^{2^s}$ in its size~$s+3$.}\label{fig:circuit}
\end{figure}

\subsection{Radical Identity testing}

Let  $f(x_1,\ldots,x_k)$ be  a multivariate polynomial computed by an algebraic circuit, and  $\mysqrt
{d_1}{a_1},\ldots,\mysqrt{d_k}{a_k}$ be  $k$ radicals, where the radicands $a_i\in \NN$, and the exponents~$d_i\in \NN$ are nonnegative integers,  written in binary.
The \emph{Radical Identity Testing} (RIT) problem asks whether 
\[f( \mysqrt{d_1}{a_1},\ldots,\mysqrt{d_k}{a_k} )=0.\]  
We define the size of an RIT instance as the maximum of
the size of the circuit  and the bit-length of the radicands $a_i$ and exponents $d_i$.

\paragraph{The $2$-RIT problem.} This is a special case of RIT 
where all input radicals $\sqrt{a_1},\ldots,\sqrt{a_k}$ are square roots and all radicands $a_i$ are rational primes, written in binary.

\paragraph{The Bounded-RIT problem.}
This variant is defined exactly as the RIT problem,
except that the input also includes an upper bound on the  degree of the circuit that is given in unary. Thus in Bounded-RIT the degree of the circuit is at most the size of the instance.

\subsection{Algebraic Number fields}

Recall that  $\alpha\in \CC$ is \emph{algebraic} if it is a root of a non-zero polynomial in $\QQ[x]$. 
The minimal polynomial of $\alpha$ (over~$\QQ$) is
the unique monic polynomial in $\QQ[x]$ (that is a polynomial with the leading coefficient~$1$) having $\alpha$ as a root.
The degree of $\alpha$ is defined to be the degree of its minimal polynomial, and denoted by $\deg \alpha$. 
 If the minimal polynomial has integer coefficients then we say that $\alpha$ is an algebraic integer.
 Given $a, d \in \NN$, the radical $\mysqrt{d}{a}$ is an algebraic integer.

An \emph{algebraic number field} (or simply number field) $K$ is a finite degree field extension of~$\QQ$, that is, $K$ is a field that contains $\QQ$ and has finite dimension when considered as a vector space over $\QQ$. 
The dimension of this vector space is called the degree of the extension and is denoted by $[K:\QQ]$. We further denote by $\OO_K$  the subring of $K$ comprised by the
algebraic integers
in $K$.
The ring $\OO_K$ is a finitely generated free abelian group. 

Given $n\in \NN$, we write $\zeta_n$ for the primitive complex $n$-th root of unity $\zeta_n = e^{\frac{2\pi i}{n}}$.
The $n$-th cyclotomic polynomial, denoted by~$\Phi_n$, is the minimal polynomial of $\zeta_n$. 
The number field $\QQ(\zeta_n)$ is an extension of $\QQ$ obtained by adjoining $\zeta_n$; the degree $[\QQ(\zeta_n):\QQ]$ of the extension is the degree of $\Phi_n$. It is well known that the ring of  integers of $\QQ(\zeta_n)$ is the ring $\ZZ[\zeta_n]$
that is generated over $\ZZ$ by $\zeta_n$.

Given a polynomial $f\in \mathbb{Q}[x]$, the splitting field 
$K$ of $f$ is the subfield of $\mathbb{C}$ generated by 
the roots of $f$.  We say that a field extension $K/\mathbb{Q}$ is a Galois extension if $K$ is the splitting field of some 
polynomial.  
The Galois group of $K$ over~$\QQ$, denoted by $\Gal(K/\QQ)$, is comprised of all automorphisms of $K$ that fix $\QQ$ pointwise. The image of $\alpha \in K$ under an automorphism  $\sigma\in \Gal(K/\QQ)$ is
called a Galois conjugate of $\alpha$, in other words, the Galois conjugates of $\alpha$ are precisely all the roots of its minimal polynomial. 
The Galois conjugates of $\zeta_n$ are all its powers~$\zeta^k_n$ with $\gcd(k,n)=1$. 
The norm of $\alpha$ over a Galois extension~$K$ is defined as the product of all the Galois conjugates of $\alpha$:
\[N_{K/\QQ}(\alpha)=\prod_{\sigma\in \Gal(K/\QQ)}\sigma(\alpha).\]
Note that the norms of all Galois conjugates are equal,
and the norm of an algebraic integer is always a rational integer itself. 

The minimal polynomial of the real radical $\mysqrt{d}{a}$ over $\QQ$ has the form
$x^{t}-c$ where $t$ is the smallest positive integer such that there exists an integer $c$ with $\mysqrt{d}{a}=\mysqrt{t}{c}$.
The conjugates of $\mysqrt{d}{a}$ are then $\zeta^j_t\mysqrt{t}{c}$ with $1\leq j \leq t$.
The splitting field $K$ of
$\prod_{i=1}^k (x^{d_i}-a_i)$ is $\QQ(\mysqrt{d_1}{a_1}, \ldots, \mysqrt{d_k}{a_k},\zeta_d)$,  obtained by adjoining the radicals $\mysqrt{d_i}{a_i}$ and  a primitive root of unity $\zeta_d$, with order $d=\lcm(d_1,\ldots,d_k)$, to~$\QQ$.

Given a number field~$K$,
the ring of integers $\OO_K$ may not be a unique factorisation domain. However we do have unique factorisation of ideals 
into products of prime ideals in $\OO_K$.  Recall here that an
ideal $\mathfrak{p} \subset \OO_k$ is called a prime ideal if for all algebraic integers $\alpha$ and $\beta$, if $\alpha\beta \in \mathfrak{p}$, then at least one of $\alpha$ and $\beta$ is in $\mathfrak{p}$. 
We say that a prime $p\in \mathbb{Z}$ splits completely in
$\OO_K$ 
if we can write $p\OO_K$ as:
\[ p\mathcal{O}_K = \mathfrak{p}_1 \cdots \mathfrak{p}_n \]
where the $\mathfrak{p}_i$ are distinct prime ideals of $\OO_K$ and $n = [K:\mathbb{Q}]$.
In the case when $K$ is the splitting field of a polynomial $f \in \mathbb{Z}[X]$
over $\mathbb{Q}$, a sufficient condition for $p$ to split completely 
over $K$ is that $f$ split into distinct linear factors over the field 
$\mathbb{Q}_p$ of $p$-adic numbers.

\subsection{Symbolic Algorithm}
Our approach to identity testing 
of algebraic numbers generalises the well-known
fingerprinting procedure for solving ACIT, which involves evaluating an arithmetic circuit modulo a randomly chosen prime. The soundness of the latter approach relies on the fact that if the integer~$z\in \ZZ$ computed by the circuit is non-zero, 
then one can with high probability randomly sample a prime $p\in \mathbb{Z}$ of size polynomial in the bit-length of the input such that $Z$ is non-zero modulo $p$.

\bigskip
\noindent {\bf Placing RIT in \coNP:}
Given input radicals~$\mysqrt{d_1}{a_1},\ldots,$ $\mysqrt{d_k}{a_k}$ to RIT, we first reduce RIT to the case that the radicands~$a_i$ are pairwise coprime numbers and that the minimal polynomials of the input radicals $\mysqrt{d_i}{a_i}$ are $x^{d_i}-a_i$ for all $i= 1,\ldots,k$.  To do this we generalise the reduction in~\cite{blomer98} and use the \emph{factor refinement} algorithm~\cite{BACH1993199}; see \Cref{appendix:reduction}.
Denote by $K$ the splitting field  of $\prod_{i=1}^k (x^{d_i}-a_i)$.

In~\Cref{fig:conp_RIT}, we present  a conceptually simple algorithm for deciding RIT and  place this problem in {\coNP}.
Our {\coNP} procedure 
 is similar in spirit to the {\coRP}
procedure for the ACIT problem, alluded to above, inasmuch as it involves
evaluating the circuit modulo a prime ideal.
 The proof of the correctness though is not as straightforward and relies on characteristics of the Galois group of~$K$ and
 concepts from number theory. 

In the case of RIT we think of the evaluation as occurring in
the ring of integers $\mathcal{O}_K$. Specifically, the idea is to work modulo a prime ideal~$\mathfrak{p}$ of
$\mathcal{O}_K$ such that the quotient $\mathcal{O}_K / \mathfrak{p}$ is a
finite field~$\mathbb{F}_p$ for some rational prime $p$.
Note that the algorithm works directly with the finite field $\mathbb{F}_p$---the prime ideal 
$\mathfrak{p}$ is implicit in the choice of radicals in 
Line~1, and ideals in $K$ only feature in the 
proof of correctness of the algorithm.   Overall, the key idea is that 
if 
$f(\mysqrt{d_1}{a_1},\ldots,\mysqrt{d_k}{a_k}) \neq 0$
then there is a polynomial-length polynomial-time checkable witness 
of this fact --- namely
a prime~$p$ and 
$\overline{\alpha}_1,\ldots,\overline{\alpha}_k \in \mathbb{F}_p$, satisfying $\overline{\alpha}_i ^ {d_i} \equiv a_i \pmod{p}$, such that $\overline{f}(\overline{\alpha}_1,\ldots,\overline{\alpha}_k )$ is non-zero, where $\overline{f}$ is the reduction of $f$ modulo~$p$.

\begin{figure*}[t]
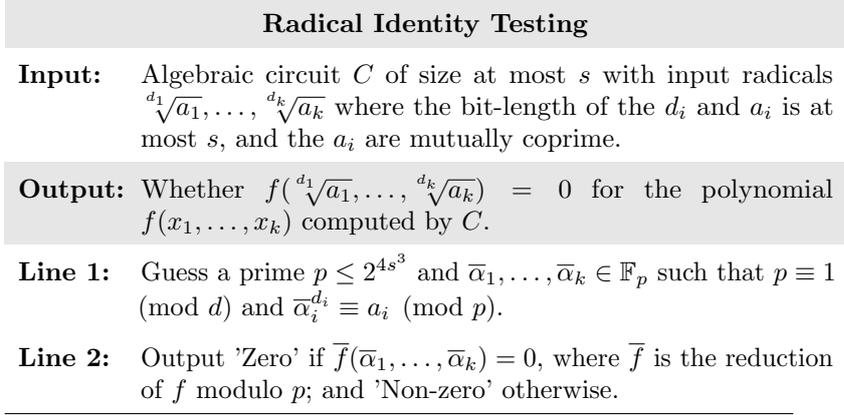

    \centering
    \begin{tabularx}{0.82\textwidth}{ p{0.1\textwidth} p{0.72\textwidth} }
   \rowcolor{Gray}
    \multicolumn{2}{c}{\textbf{Radical Identity Testing}} \\
    \textbf{Input:} & Algebraic circuit $C$
    of size at most $s$ with input radicals $\mysqrt{d_1}{a_1},\ldots,\mysqrt{d_k}{a_k}$ where
    the bit-length of the $d_i$ and $a_i$ is at most $s$,
    and the $a_i$ are mutually coprime. 
    \\  \rowcolor{Gray}
    \textbf{Output:} & Whether $f(\mysqrt{d_1}{a_1},\ldots,\mysqrt{d_k}{a_k})=0$ for the polynomial $f(x_1,\ldots,x_k)$ computed by $C$.
    \\  
        \textbf{Line 1:} & 
            Guess a prime $p \leq 2^{4s^3}$ and $\overline{\alpha}_1,\ldots,\overline{\alpha}_k \in \FF_p$ such that $p \equiv 1 \pmod{d}$ and $\overline{\alpha}_i^{d_i} \equiv a_i \pmod{p}$. 
    \\ 
    \textbf{Line 2:} &
            Output 'Zero' if $\overline{f}(\overline{\alpha}_1,\ldots,\overline{\alpha}_k) = 0$, where $\overline{f}$ is the reduction of $f$ modulo $p$; and 'Non-zero' otherwise.
    \\ \hline
    \end{tabularx}
    \caption{Our nondeterministic algorithm for the complement of RIT.}
    \label{fig:conp_RIT}
\end{figure*}

We first summarise the algorithm, see \Cref{fig:conp_RIT}, and then outline the correctness argument.

\begin{enumerate}
      \item We find a rational prime $p$ such that
    each polynomial among $x^{d_1}-a_1, \ldots, x^{d_k} - a_k$ splits into 
    distinct linear factors over $\mathbb{F}_p$.  We can find such a prime $p$ in non-deterministic polynomial time
    in the length of the problem instance by (i)~choosing a prime~$p$ such that $p\equiv 1 \bmod d$, which 
    ensures that $\mathbb{F}_p$ contains a primitive $d$-th root of unity,    
and (ii)~guessing and checking 
    $\overline{\alpha}_1,\ldots,\overline{\alpha}_k \in \mathbb{F}_p$ such that 
    $\overline{\alpha}_i$ is any root in $\mathbb{F}_p$ of the polynomial $x^{d_i}-a_i$.

    \item We evaluate the polynomial $\overline{f}(\overline{\alpha}_1,\ldots,\overline{\alpha}_k)$ in $\mathbb{F}_p$, where $\overline{f}(x_1,$ $\ldots,x_k) \in \mathbb{F}_p[x_1,\ldots,x_k]$ is the reduction of $f$ modulo $p$.
    If the result  of this computation is zero then we report 'Zero'; otherwise we report 'Non-zero'. 
    \end{enumerate}

    Having summarised the algorithm we can now state our main result: 
 
\begin{theorem}\label{th:rit-conp}
	The RIT problem is in {\coNP} under GRH.
\end{theorem}
    
 Below, we give a high-level idea of the correctness of~\Cref{th:rit-conp}.
   The first element of the correctness proof of the {\coNP} algorithm  is to argue that the prime $p$ chosen in Item~1 completely splits
    in the ring of integers~$\mathcal{O}_K$.
    In this situation, for any prime-ideal factor $\mathfrak{p}$ of $p$, each quotient field $\mathcal{O}_K / \mathfrak{p}$
    is isomorphic to the finite field $\mathbb{F}_p$.  
    By standard results in algebraic number theory we know that $p$ completely splits in $\mathcal{O}_K$
    if each polynomial $x^{d_1}-a_1,\ldots,x^{d_k}-a_k$ splits into linear factors over the field $\mathbb{Q}_p$ of $p$-adic numbers.
    But the latter requirement is guaranteed by Hensel's Lemma in tandem with Conditions 1(i) and 1(ii) that determine the choice of $p$ in the algorithm.
    In more detail, 
    Condition 1(i), that $p \equiv 1 \bmod d$, entails that $\mathbb{F}_p$ contains 
    a primitive $d$-th root of unity. 
Indeed, since the powers of a root of unity are also all  roots of unity themselves, and since the multiplicative group $\FF_p^{*}$ is cyclic, it is clear that $\FF_p^{*}$ contains
a primitive $d$-th root of unity just in case $d \mid p-1$. 
      In combination with Condition 1(ii), that each of the polynomials 
    $x^{d_1}-a_1, \ldots, x^{d_k} - a_k$
    has a root in $\mathbb{F}_p$, we can conclude that
    each of the above polynomials in fact splits into distinct linear factors over $\mathbb{F}_p$.
    Then Hensel's Lemma allows us to lift this factorisation over $\mathbb{F}_p$ into a factorisation 
    into distinct linear factors over $\mathbb{Q}_p$.

    The second element of the correctness proof concerns the choice of 
    $\overline{\alpha}_1,\ldots,\overline{\alpha}_k$ in $\mathbb{F}_p$.   In particular,
    we argue that the correctness of the algorithm does not rely, for $i=1,\ldots,k$, on a specific choice of $\overline{\alpha}_i$
    among the $d_i$ roots of $x^{d_i} - a_i$ in $\mathbb{F}_p$.  This argument is based on the fact
    that the Galois group
    $\mathrm{Gal}(K / \mathbb{Q})$ acts transitively on the set 
    \[ \{ (\alpha_1,\ldots,\alpha_k) \in K^k : \alpha_1^{d_1}=a_1 \wedge \cdots \wedge \alpha_k^{d_k} = a_k \} \]
   	   	 That is, for any two $k$-tuples $\tau_1$ and $\tau_2$ in the set above, there exists an automorphism $g$ in $\Gal(K/\QQ)$ such that $g(\tau_1) = \tau_2$.
    We will call this property joint transitivity; see \Cref{lem:jtrans}.

     Now for every prime ideal factor $\mathfrak{p}$ of $p\mathcal{O}_K$ there is a
surjective homomorphism $\varphi: \mathcal{O}_K \rightarrow \mathbb{F}_p$ with kernel
$\mathfrak{p}$.   For each choice of $\mathfrak{p}$ and for all $i=1,\ldots,k$, the corresponding homomorphism 
$\varphi$ maps $\alpha_i$ to some root $\overline{\alpha}_i$ of $x^{d_i} - a_i$ in $\mathbb{F}_p$.
Conversely, using joint transitivity, we are able to show that every
mapping $\alpha_1 \mapsto \overline{\alpha}_1, \ldots, \alpha_k \mapsto \overline{\alpha}_k$
where, for $i=1,\ldots,k$, $\overline{\alpha}_i$ is an arbitrary root of $x^{d_i} - a_i$ in $\mathbb{F}_p$,
arises from the quotient map by some prime ideal factor of $p$.
We conclude that the value $\overline{f}(\overline{\alpha}_1,\ldots,\overline{\alpha}_k)$ in Item 2 is 
the image of $f(\alpha_1,\ldots,\alpha_k)$ under the quotient map $\OO_K \to \OO_K / \mathfrak{p}$ for some prime ideal $\mathfrak{p}$.
    
    It follows from the line immediately above that
    the algorithm has no false positives: if $f(\alpha_1,\ldots,\alpha_k)=0$ in $K$ then certainly 
    $\overline{f}(\overline{\alpha}_1,\ldots,\overline{\alpha}_k)= 0$.
    We moreover show that for a suitable choice of prime $p$, namely such that $p$ does not
    divide the norm of~$f(\alpha_1,\ldots,\alpha_k)$ over $K/\mathbb{Q}$,
    the converse holds:
    if $f(\alpha_1,\ldots,\alpha_k)\neq 0$ in $K$
    then
    $\overline{f}(\overline{\alpha}_1,\ldots,\overline{\alpha}_k)\neq 0$
    in $\mathbb{F}_p$.
For more details, see \Cref{th:RIT-soundess}.

A quantitative version of the Chebotarev density theorem
guarantees that we can find such a prime of polynomial bit-length that splits.
Informally speaking, Chebotarev's density theorem states that the set of rational primes $p$ that split completely in $\OO_K$ has density $\frac{1}{|\Gal(K/\QQ)|}$.
The original statement of the theorem~\cite{Stevenhagen95chebotarevand}
is asymptotic, whereas for our algorithm we use its quantitative version in order to obtain a bound on the number of primes $p$ of size polynomial in the bit-length of the input such that $p\OO_K$ completely splits, which requires GRH~\cite{lagariasO, serre1981quelques}.
We use the bound in combination with a bound on the norm of $f(\alpha_1,\ldots,\alpha_k)$ to ensure we can find at least one small prime that will not divide the norm and ensure our computation is sound.

\bigskip
\noindent {\bf Placing $2$-RIT in {\coRP}:}
We improve the obtained {\coNP} bound for  RIT  in the case of $2$-RIT, wherein all input radicals are square roots and all radicands $a_i$ are 
rational odd primes
(written in binary).
We show that $2$-RIT is in {\coRP} under GRH, and that it is in {\coNP} unconditionally.
Our {\coRP} algorithm  chooses a suitable random prime to be used in the symbolic computation. 
As discussed above,  the natural density of such suitable primes is $\frac{1}{|\Gal(K/\QQ)|}$, which is not  sufficient even for $2$-RIT.  However, we show that there is an arithmetic progression with a good density of primes, and that all primes in this progression are suitable.
To obtain this result, we rely on the law of quadratic reciprocity, as well as Dirichlet’s theorem on the density of primes in arithmetic progressions.

We recall that a suitable  prime $p$ for our symbolic algorithm  is  such that 
the minimal polynomials of all input radicals split into linear factors over~$\FF_p$. In the setting of $2$-RIT, the minimal polynomials of the inputs are in the form $x^2-q_i$ where $q_i$ is a rational prime. The above requires that the equations $x^2 \equiv q_i$ have  solutions in $\FF_p$, that is, $q_i$ is a quadratic residue modulo~$p$. 
But then by the law of quadratic reciprocity, $p$ is a quadratic residue modulo prime $q_i$ if and only if $q_i$ is a quadratic residue modulo~$p$, condition to  $p \equiv 1 \pmod{4}$. 
Roughly speaking, the latter holds if $p \equiv 1 \pmod{4q_i}$ (as $1$ is a perfect square in $\FF_p$). 

By the Chinese remainder theorem and a more detailed argument similar to the above, 
we show 
that there is an arithmetic progression $A\NN+b$ such that 
for all primes $p$ in the progression, all  polynomials 
$x^2-q_i$, with $i\in \{1,\cdots,k\}$, split into linear factors over~$\FF_p$. 
We further impose another
condition on $A$
and $b$, based on Pocklington's algorithm, such that 
a root of each $x^2-q_i$ can be computed in deterministic polynomial time in the length of the problem instance.   

Finally, 
we use  estimates on the density of primes in an arithmetic
progression, see~\Cref{th:prime_density}, to obtain the following result:

\begin{theorem}\label{th:squareRIT-coRP}
The $2$-RIT problem is in {\coRP} assuming GRH and in {\coNP} unconditionally.
\end{theorem}

\section{The RIT algorithm}
\label{sec:conp}

In this section, we present a nondeterministic polynomial time algorithm for the complement of RIT. As discussed in \Cref{sec:overview}, the idea is to work in a finite field obtained by quotienting the ring of integers of the splitting field of the input radicals by a suitable prime ideal.

Below, we fix an instance of  RIT  given by an algebraic circuit $C$, and input radicals~$\mysqrt{d_1}{a_1},\ldots,\mysqrt{d_k}{a_k}$ with respective minimal polynomials $x^{d_i}-a_i$, where the radicands~$a_i$ are pairwise coprime;
this assumption is without loss of generality as discussed in \cref{appendix:reduction}.
We denote by $s$ the size of our fixed RIT instance; that is, the size of the circuit is bounded by $s$,  and the magnitude of the $a_i$ and $d_i$ is at most $2^s$. Note that  $k\leq s$, by the definition of size of an algebraic circuit. 
 We further  denote by $K$ the splitting field  of $\prod_{i=1}^k (x^{d_i}-a_i)$, which can be generated by adjoining to~$\QQ$ the radicals $\mysqrt{d_i}{a_i}$ and a primitive $d$-th root of unity~$\zeta_d$, with $d = \lcm(d_1,\ldots,d_k)$. We denote by $\OO_K$ the ring of integers of $K$.

In our construction, we evaluate the polynomial given by an algebraic circuit in a finite field $\FF_p$ for some rational prime $p$ that splits completely in $\OO_K$, that is, such that $p\OO_K = \mathfrak{p}_1\cdots \mathfrak{p}_n$ where $\mf{p}_i$ are distinct prime ideals of $\OO_K$, and $n$ is the degree of the number field $K$.

In general, given a number field $L$ and a rational prime $q$ with prime ideal $\mathfrak{q}$ dividing $q\OO_L$, we say that $\mathfrak{q}$ lies above $q$ in $\OO_L$. The residue field $\OO_L/\mathfrak{q}$ is isomorphic to an extension of the finite field~$\FF_q$, 
and we have that $\mathfrak{q} \cap \ZZ=q\ZZ$.
 However, in our special case of a completely split prime $p$ in $\OO_K$, all residue fields $\OO_K/\mathfrak{p}_i$ are isomorphic to $\FF_p$. This is crucial, as it ensures that the values in our finite field computation really will all be in $\FF_p$ and not in one of its finite extensions (see the discussion in \Cref{sec:discussion}).

The following proposition asserts that a prime $p$ completely splits in $K$ if the minimal polynomials $x^{d_i}-a_i$ of our input radicals and the $d$-th cyclotomic polynomial (the minimal polynomial of $\zeta_d$) split into distinct linear factors in $\FF_p$.

\begin{proposition}
	\label{prop:split_primes_qp}
	Given a monic polynomial $g \in\ZZ[x]$ and its splitting field~$L$, 
	a prime $q\in\ZZ$ splits completely in $L$ if $g$ splits into distinct linear factors in $\FF_q$. 
\end{proposition}

\begin{proof}[Proof sketch]
Denote by $L_{\mathfrak{q}}$ the finite field extension of $\QQ_q$ obtained by adjoining the roots of $g$ to $\QQ_q$. Let $\mathfrak{q}$ be a prime ideal lying above $q$ in $L$.
	Recall that the decomposition group of a prime ideal $\mathfrak{q} \subset L$ is defined as the set of all automorphisms of $\Gal(L/\QQ)$ that fix $\mathfrak{q}$, i.e. $D_{\mf{q}} = \{ \sigma \in \Gal(L/\QQ) \mid \sigma(\mathfrak{q}) = \mathfrak{q}\}$. Furthermore, since the field $L$ is Galois, the following isomorphism holds:
\begin{equation}
\label{eq:decomposition_group_iso}
	D_{\mathfrak{q}} \cong \Gal(L_{\mathfrak{q}}/\QQ_q)
\end{equation}

Now given that $g$ splits into distinct linear factors in $\FF_q$, by virtue of $L$ being Galois, it follows that all roots of $g$ in $\FF_q$ are non-zero. Hence we can use Hensel's lemma to construct the solutions of $g$ in $\QQ_q$, that is, $g$ splits completely in $\QQ_q$. This, in turn, implies that $L_{\mathfrak{q}} = \QQ_q$, that is, $\Gal(L_{\mathfrak{q}}/\QQ_q)$ is trivial and \eqref{eq:decomposition_group_iso} asserts that the same holds for $D_{\mathfrak{q}}$. This entails that the only automorphism that fixes the prime ideal $\mathfrak{q}$ is the identity, whereas for all non-trivial $\sigma\in\Gal(L/\QQ)$, which permute the conjugates of our input radicals, they will map elements of $\mathfrak{q}$ to elements of $\OO_L/\mathfrak{q}\cong \FF_q$, i.e., $\FF_q$ will indeed contain all roots of $g$.
\end{proof}

\begin{example}
Consider an RIT instance asking whether the polynomial $f(x) = x^2 - 10$ vanishes at the radical input $\sqrt{5}$. The computation occurs in the field $L=\QQ(\sqrt{5})$, with ring of integers $\OO_L = \ZZ[\frac{1+\sqrt{5}}{2}]$.
	We can observe that $11$ is a completely split prime and note that the principal ideal of $\OO_L$ generated by~$11$ factors as $11\OO_L = (4+\sqrt{5}) (4-\sqrt{5})$.
Since  $11$ totally splits in~$\OO_L$, we have $\OO_L/\mathfrak{q} \cong \FF_{11}$ for the prime ideal $\mathfrak{q}=(4+\sqrt{5})$ lying above~$11$.
The polynomial $x^2-5$  completely splits to  $(x-4)(x+4) $ in $\FF_{11}$, and    consequently we have that $\QQ_{11}(\sqrt{5}) = \QQ_{11}$. 
	The rational prime~$5$, however, is an example of the primes that we want to avoid as $5\OO_L = (\sqrt{5})^2$.
	 In particular, the polynomial $x^{2}-5$ is irreducible in~$\FF_{5}$, implying that $L_{\mathfrak{q}} = \QQ_q(\sqrt{5})$ and $[L_{\mathfrak{q}} : \QQ_q] = 2$. If we evaluate $f$ on the input $\sqrt{5}$ in $\FF_{11}$, we get a true negative, as the computed value will be $6\in\FF_{11}$, whereas evaluating the polynomial in $\FF_5$ would give us a false positive.
\end{example}

\subsection{Proof of correctness}
\label{sec:correctness}

Given a polynomial~$g\in \ZZ[x]$ that is irreducible over~$\QQ$, the Galois group $\Gal(L/\QQ)$ of the splitting field~$L$ of $g$ acts transitively on the roots of $g$~\cite[Proposition~22.3]{StewartBook} . We show that a stronger notion of transitivity holds for our real radicals~$\mysqrt{d_i}{a_i}$:  

\begin{lemma}\label{lem:jtrans}
The group
$\mathrm{Gal}(K/\mathbb{Q})$ acts transitively
on the set of $k$-tuples
\[ \mathcal{S}:=\left\{(\alpha_1,\ldots,\alpha_k) \in K^k \mid 
\alpha_1^{d_1}=a_1 \wedge \cdots \wedge \alpha_k^{d_k}=a_k 
\right\} \]
 \end{lemma}
\begin{proof}
Recall that $\mysqrt{d_1}{a_1},\ldots,\mysqrt{d_k}{a_k}$ are real radicals with respective minimal polynomials $x^{d_i}-a_i$ and~$a_i$ mutually coprime.

Let $L_i := \QQ(\mysqrt{d_1}{a_1}, \ldots, \mysqrt{d_i}{a_{i}})$ for $i\in\{1,\ldots,k\}$.
By virtue of the~$a_i$ being coprime and~\cite[Lemma 4.6]{blomer98}, 
the polynomial~$f_i := x^{d_i}-a_i$ stays irreducible over~$L_{i-1}$, and thus is the minimal polynomial of $\mysqrt{d_i}{a_i}$
 over this field.
 
The proof follows by repeated use of the Isomorphism extension theorem, cf. \cite[Theorem~5.12]{StewartBook}.
Note that $L_{i}$ is a simple extension of $L_{i-1}$ with $L_{i} = L_{i-1}(\alpha_i)$ where $\alpha_i$ is a solution of $f_i$. Denote by $\psi_{i-1}$ an embedding of $L_{i-1}$ into $K$. If $\alpha_i'$ is a root of $\psi_{i-1}(f_i)$ then there is a unique extension of $\psi_{i-1}$ to a homomorphism $\psi_i : L_i \to K$ such that $\psi_i(\alpha_i) = \alpha_i'$.
Applying the above inductively, we obtain a homomorphism $\psi_{k} : \QQ(\mysqrt{d_1}{a_1}, \ldots, \mysqrt{d_k}{a_{k}}) \to K$, which by the Isomorphism extension theorem can again be extended to an automorphism of $K$ acting jointly transitively on $\mathcal{S}$.
\end{proof}

We now show that for an instance of  RIT, that is, an algebraic circuit with underlying polynomial $f$, and radical input with pairwise coprime radicands, the finite field computation is sound.
Given a rational prime~$p$ and a polynomial $g(x)\in \ZZ[x]$ we denote by $\bar{g}\in \FF_p[x]$  the reduction of $g$ modulo $p$.

\begin{lemma}\label{th:RIT-soundess}
Let $p$ be a prime that completely splits in $\OO_K$, and
let $\overline{\alpha}_1,\ldots,\overline{\alpha}_k \in \mathbb{F}_p$
be roots of the polynomials $x^{d_1}-a_1,\ldots,x^{d_k}-a_k$, respectively.
Then for all 
 $f(x_1,\ldots,x_k)\in\ZZ[x_1,\ldots,x_k]$, we have
\begin{enumerate}
	\item if $f(\mysqrt{d_1}{a_1},\ldots,\mysqrt{d_k}{a_k}) = 0$ then $\overline{f}(\overline{\alpha}_1,\ldots,\overline{\alpha}_k) = 0$, and
	\item if $\overline{f}(\overline{\alpha}_1,\ldots,\overline{\alpha}_k) = 0$ then $p\mid N_{K/\QQ}(f(\mysqrt{d_1}{a_1},\ldots,\mysqrt{d_k}{a_k}))$.
\end{enumerate}
\end{lemma}
\begin{proof}
Recall that the radicals $\mysqrt{d_1}{a_1},\ldots, \mysqrt{d_k}{a_k}$ are such that their respective minimal polynomials are $x^{d_1}-a_1,\ldots,x^{d_k}-a_k$.

Consider a prime-ideal factor $\mathfrak{p}$ of $p\OO_K$.
The quotient homomorphism $\OO_K \rightarrow \mathbb{F}_p$
with kernel $\mathfrak{p}$ maps the $d_i$ distinct roots of 
each polynomial $x^{d_i}-a_i$ in $\OO_K$ bijectively onto the 
roots of the same polynomial in $\mathbb{F}_p$.
Then, by joint transitivity of the Galois group $\mathrm{Gal}(K/\mathbb{Q})$, established in~\Cref{lem:jtrans}, there is a homomorphism
$\varphi:\OO_K\rightarrow \mathbb{F}_p$ such that 
$\varphi(\mysqrt{d_i}{a_i})=\overline{\alpha}_i$ for all $i \in \{1,\ldots,k\}$.

Item 1 follows from the fact that if $f(\mysqrt{d_1}{a_1},\ldots,\mysqrt{d_k}{a_k}) =0$
then $\overline{f}(\overline{\alpha}_1,\ldots,\overline{\alpha}_k)=\varphi(f(\sqrt[d]{a_1},\ldots,\sqrt[d]{a_k}))=0$.
For Item 2, note that
the kernel of $\varphi$ is a prime ideal of $\OO_K$ lying above $p$.
Thus if $\overline{f}(\overline{\alpha}_1,\ldots,\overline{\alpha}_k)=0$
then $p\mid N_{L/\QQ}(f(\alpha_1,\ldots,\alpha_k))$.
\end{proof}

We have just shown that given an instance of  RIT, the computation can be taken to a finite field $\FF_p$ for some rational prime $p$. In the following section, we discuss how to choose an appropriate prime~$p$ that satisfies the two conditions given in \Cref{th:RIT-soundess}.

\Cref{lem:jtrans} plays an important role in the construction of our algorithm, and furthermore is one of the properties of the input to  RIT  that makes our technique difficult to generalise to more general identity testing problems. In particular, joint transitivity ensures that in \Cref{th:RIT-soundess}(1), no matter which representative of the $\mysqrt{d_i}{a_i}$ we choose in $\FF_p$, that is, no matter which solution of the equation $x^{d_i}-a_i$ we guess in $\FF_p$, the computation remains sound. If we were, for instance, to generalise our identity testing problem to allow radical and cyclotomic inputs, joint transitivity may not hold anymore.

\begin{example}
\label{ex:generalisation}
Consider the polynomial $f(x_1,x_2) = x_2^2 - x_1 x_2 + 1$ with input $\sqrt{2}$ and a primitive $8$-th root of unity $\zeta_8$. The polynomial $f$ vanishes at $(\sqrt{2},\zeta_8)$. The number field of the computation is $\QQ(\sqrt{2},\zeta_8)$, and we can choose the completely split prime $17$ for our finite field computation. The minimal polynomials of our input split as $x^2 - 2 = (x-6)(x+6)$ and $x^4 + 1 = (x+2)(x-2)(x+8)(x-8)$ in $\FF_{17}$. However, since the Galois group of the field $\QQ(\sqrt{2},\zeta_8)$ does not act jointly-transitively on the input, we cannot choose the representatives of our two input numbers in~$\FF_{17}$ arbitrarily. In particular, by choosing $6$ for $\sqrt{2}$ and $2$ for $\zeta_8$, and evaluating~$f$ in~$\FF_{17}$, the result would be $10$, a clear false negative. This is due to the fact that the minimal polynomial $\Phi_8(x) = x^4 + 1$ of $\zeta_8$ reduces in $\QQ(\sqrt{2})$, hence, as soon as we choose $6$ as a representative for $\sqrt{2}$, we cannot choose the representative for $\zeta_8$ freely. In fact, if we replace $x_1$ by $\sqrt{2}$ and $x_2$ by $x$ in the polynomial~$f$, we obtain $(x^2 - \sqrt{2}x + 1)$,
 which is a factor of $\Phi_8(x)$ in $\QQ(\sqrt{2})$. Indeed, $\Phi_8(x)$ factors as $(x^2 - \sqrt{2}x + 1)(x^2 + \sqrt{2}x + 1)$ in $\QQ(\sqrt{2})$, hence, as soon as we choose $6$ as a representative for $\sqrt{2}$, we can only choose the roots of
 $(x^2 - 6x + 1)$
 as representatives for $\zeta_8$, whereas $2$
 is a root of the polynomial $(x^2 + 6x + 1)$ in $\FF_{17}$.
\end{example}

We note here that the underlying approach (working modulo prime ideals) is the same as in \cite{balaji2021cyclotomic} and preceding works on versions of ACIT.
However, as shown in \Cref{ex:generalisation}, the transitivity condition given in 
\Cref{lem:jtrans} is crucial in our approach, whereas it has no equivalent in \cite{balaji2021cyclotomic}.
Furthermore, we have to allow for the fact that the ring of integers of the number field is no longer monogenic in the present case (e.g., \Cref{th:RIT-soundess} is analogous to Theorem~8 in \cite{balaji2021cyclotomic}, but the proof requires working in a certain order in the ring of integers). 
In relation to this, \Cref{ex:generalisation} explains some of the difficulties arising from attempting a common generalisation of \cite{balaji2021cyclotomic} and the present paper.

\subsection{Choice of the prime $p$}

Let us call the primes $p$ that completely split in $\OO_K$
\emph{good} primes, and the primes that do not divide the norm of the algebraic integer computed by the circuit \emph{eligible} primes. We will first discuss how to choose eligible primes.
All missing proofs of Lemmas we establish in this subsection can be found in \Cref{appendix:conp}.

Given a rational prime $p$, we would like to ensure that $p \nmid N(\alpha)$, where $\alpha$ is the algebraic integer computed by a circuit.
The norm $N(\alpha)$ has at most $\log |N(\alpha)|$ prime divisors. Intuitively, this means that if we have a set of $\log |N(\alpha)| + 1$ primes, at least one of them will be eligible. Now let us see how we can bound this value.

Recall that the norm of an algebraic number $\alpha$ in a Galois field can be computed as the product of its Galois conjugates. Thus, in order to bound the magnitude of the norm of our computed number, we need a bound on the magnitude of the conjugates of $\alpha$, as well as the size of the Galois group, that is, the number of conjugates of $\alpha$. The latter can be obtained using a bound on the degree of the number field $K$ itself.
Recall that  the splitting field~$K$ of $\prod_{i=1}^k (x^{d_i}-a_i)$ is $\QQ(\mysqrt{d_1}{a_1}, \ldots, \mysqrt{d_k}{a_k},\zeta_d)$. Since the $d_i$ are of magnitude at most 
$2^s$, it follows that $d$ is of magnitude at most $2^{s^2}$,
and the degree of $K$ will be at most $2^{2s^2}$.
Using this bound, we show:

\begin{restatable}[Bound on the norm]{lemma}{boundnorm}\label{lem:norm}
Denote by $\alpha\in \OO_{K}$ the algebraic integer computed by~$C$ evaluated on the~$\mysqrt{d_i}{a_i}$. We have
\[ |N(\alpha)| \leq 2^{2^{s^3}} \]
for $s\geq 4$.

\end{restatable}

In the context of our algorithm, this means that if we find 
$2^{s^3} + 1$ good primes, at least one of them will also be eligible, and hence our finite field computation will be sound. Let us now focus on how we can ensure to be able to find enough good primes of polynomial bit-length in the size of the input to complete our reasoning.

To this aim, we use a quantitative version of the Chebotarev density theorem. Intuitively speaking, given a Galois extension $L$ of $\QQ$, the theorem gives a bound on the number of primes splitting in a certain pattern in $\OO_L$. The different classes of splitting patterns correspond to conjugacy classes of the Galois group $\Gal(L/\QQ)$ of $L$. As it turns out, the primes splitting completely in $L$ correspond to the conjugacy class $\{id\}$  
containing solely the identity element~$id$ of $\Gal(L/\QQ)$. The asymptotic version of the theorem then asserts that the set of completely split primes has density $\frac{1}{|\Gal(L/\QQ)|}$.  Denoting by $\pi(x)$ the number of all rational primes less or equal to $x$, and by $\pi_1(x)$ the number of completely split primes, the quantitative version of the theorem is as follows~\cite{lagariasO, serre1981quelques}:

\begin{proposition}[Bound on $\pi_{1}(x)$]\label{prop:number_of_primes}
Assuming GRH,
\[ \pi_{1}(x) \geq \frac{1}{|\Gal(K/\QQ)|} \left[ \pi(x) - \log \Delta_K - c x^{1/2} \log (\Delta_K x^{|\Gal(K/\QQ)|}) \right] \]
where $c$ is an effective constant.
\end{proposition}

To apply the above proposition we will need a bound on the discriminant~$\Delta_K$ of the number field $K$. Recall the definition of the discriminant of a number field; given a $\ZZ$-basis $\{\alpha_1,\ldots,\alpha_n\}$ of the ring of integers $\OO_L$ of a number field $L$, the discriminant $\Delta_L$ is the determinant of the matrix $\mathrm{tr}(\alpha_i\alpha_j)$ for all $i,j=1,\ldots,n$.
However, in the case of a radical field extension $L$, 
 computing a basis for its ring of integers~$\OO_L$ is a non-trivial task. In order to avoid this computation, we try to bound the discriminant using the discriminant of an order of our ring of integers $\OO_K$. Recall that an order~$\OO$ in a number field $L$
is a free $\ZZ$-submodule of $\OO_L$ of rank $[L:\QQ]$. 
Looking again at the number field $L=\QQ(\sqrt{5})$ with ring of integers $\OO_L=\ZZ[\frac{\sqrt{5}+1}{2}]$, note that $\ZZ[\sqrt{5}] \subset \OO_L$; in particular, $\ZZ[\sqrt{5}]$ is an order of index 2 in $\OO_L$. The following (see, e.g., \cite[Proposition 4.4.4]{cohen2013course}) holds:

 \begin{proposition}
 \label{prop:order_disc}
 Suppose $\OO$ is an order in $\OO_L$. Then
 \[ \Disc(\OO) = \Disc(\OO_L) \cdot [\OO_L : \OO]^2 \]
 \end{proposition}
 
We construct an order $\OO$ of $\OO_K$, the discriminant of which we can bound using a standard result in algebraic number theory. The Primitive element theorem states that any number field $L$ can be generated by adjoining a single element $\theta$, called the primitive element, to $\QQ$, i.e., $L = \QQ(\theta)$. Then the subring $\ZZ[\theta]$ of $\OO_L$ is an order of $\OO_L$. Furthermore, the discriminant $\Disc(\ZZ[\theta])$ is equal to the discriminant of the minimal polynomial of $\theta$.
It is well-known that the discriminant of a polynomial can be computed as the product of the differences of the roots.

We thus start with preliminary result regarding primitive elements of the number field~$K$.

\begin{restatable}[Bound on the primitive element]{lemma}{boundprimitiveelement}\label{lem:primitive_el}
The field $K$ has a primitive element $\theta$, computed as the linear combination
\[ \theta = c_0 \zeta_d + \sum_{i=1}^k c_i \mysqrt{d_i}{a_i}\]
with $c_i \leq 2^{4s^2} \in \ZZ$ and $\deg \theta \leq 2^{2s^2}$.
\end{restatable}

Henceforth, we fix a primitive element~$\theta$ for our number field~$K$, computed as in Lemma~\ref{lem:primitive_el}. Now \Cref{prop:order_disc}
suggests that $\Delta_K \leq \Delta_{f_\theta}$, where $\Delta_{f_\theta} = \Disc(\ZZ[\theta])$, which we bound as follows:

\begin{restatable}[Bound on the discriminant]{lemma}{bounddiscriminant}\label{lem:discriminant}
 We have
 \[ |\Disc(\ZZ[\theta])| \leq 2^{2^{5s^2}}\]
 for $s\geq 4$.
\end{restatable}

Recall that we would like to choose enough good primes $p$ so that at least one of them will be eligible, i.e., at least one of them will not divide the norm of the computed algebraic integer. In particular, we would like to find $x$ such that
 $\pi_{1}(x) \geq 2^{s^3} + 1$.
Using the bound above, we claim that this is the case for 
$x \geq 2^{4s^3}$:

\begin{restatable}{lemma}{boundnprimes}\label{prop:bound_p_conp}
Assuming GRH,
\[ \pi_{1}(2^{4s^3}) \geq  2^{s^3}  + 1.\]
\end{restatable}

\subsection{Proof of \Cref{th:rit-conp}}

\begin{proof}[Proof of Theorem~\ref{th:rit-conp}]
	\Cref{fig:conp_RIT} presents a nondeterministic polynomial time algorithm for the complement of  RIT  as follows.
	
	Given input radicals $\mysqrt{d_1}{a_1},\ldots,\mysqrt{d_k}{a_k}$,
	denote by $K$ the splitting field  of $\prod_{i=1}^k (x^{d_i}-a_i)$.
	 Further denote by $\theta$ a primitive element of~$K$, computed as in Lemma~\ref{lem:primitive_el}. 
	
	Let us first argue that the algorithm runs in polynomial time. 
In Step 1, after guessing candidates for $p$ such that $p \equiv 1 \pmod{d}$ and $\overline{\alpha}_1,\ldots,\overline{\alpha}_k$,  verifying whether $\overline{\alpha}_i^{d_i} \equiv a_i \pmod{p}$ can be done in polynomial time by the repeated-squaring method. It is clear that Step 2 can be done in polynomial time.
	
	Now let us show that the RIT problem is in \coNP. Suppose $f(\mysqrt{d_1}{a_1},\ldots,\mysqrt{d_k}{a_k}) \neq 0$. Under GRH, the lower bound in \Cref{prop:bound_p_conp} shows that
	 $\pi_{1}(2^{4s^3}) \geq 2^{s^3} + 1$.
	 It follows that there exists a prime
	 $p\leq 2^{4s^3}$  such that 
	\begin{itemize}
		\item $p \nmid N(f(\mysqrt{d_1}{a_1},\ldots,\mysqrt{d_k}{a_k}))$, and
		\item $p$ splits completely in $K$.
	\end{itemize}
	
	The polynomial certificate of non-zeroness of $f(\mysqrt{d_1}{a_1},\ldots,\mysqrt{d_k}{a_k})$ then comprises, the prime $p$ above, as well as the list of integers $\overline{\alpha}_1,\ldots,\overline{\alpha}_k\in\FF_p$ such that $\overline{\alpha}_i^{d_i}\equiv a_i \pmod{p}$. Following \Cref{th:RIT-soundess}, we then have that $\overline{f}(\overline{\alpha}_1,\ldots,\overline{\alpha}_k) \neq 0$.
	   
	   On the other hand, as we have noted above, for any prime~$p$ and the representation~$\overline{\alpha}_1,\ldots,\overline{\alpha}_k$ of radicals $\mysqrt{d_1}{a_1},\ldots,\mysqrt{d_k}{a_k}$ in $\FF_p$, if 
	   $f(\mysqrt{d_1}{a_1},\ldots,\mysqrt{d_k}{a_k}) = 0$, then $\overline{f}(\overline{\alpha}_1,\ldots,\overline{\alpha}_k) = 0$, as shown in \Cref{th:RIT-soundess}, which concludes the proof.
\end{proof}

\section{The $2$-RIT algorithm}
\label{sec:corp}
Our algorithm for RIT uses non-determinism to guess a "good" prime.  It is natural to wonder whether such primes can instead be randomly sampled.
Recall we require that the prime $p$ have polynomial bit-length in the size of the input, and that  the congruences $x^{d_i} \equiv a_i \pmod{p}$ are solvable in $\FF_p$. By Chebotarev's density theorem, roughly speaking, the density of such good primes is
$\frac{1}{|\Gal(K/\QQ)|}$.
Since the size of the Galois group of $K$ over $\QQ$ is exponential in the size of the input, good primes do not have sufficient density in order to directly be chosen randomly. 
The density remains insufficient even if the exponents~$d_i$ are prime numbers written in unary.

In what follows, we show that   $2$-RIT is in {\coRP} under GRH, and  in {\coNP} unconditionally. Recall that  $2$-RIT  is the identity testing problem for an algebraic circuit~$C$ evaluated on square-roots
$\sqrt{a_1},\ldots,\sqrt{a_k}$ for 
$k$ rational odd primes $a_1,\ldots,a_k$.
The proofs from \Cref{sec:correctness} ensure that the finite field computation in our algorithm is sound; it remains to show how to choose a completely split prime~$p$ and determine the solutions to the equations $x^{2} \equiv a_i \pmod{p}$  in~$\FF_p$.

As noted above, 
the natural density of primes is not sufficient even for  $2$-RIT, 
however, we show that there is an arithmetic progression with a good density of primes, and that all primes in this progression are good.

Below, we state known effective bounds on the density of primes in an arithmetic progression, in particular, the following estimates, which have been shown in \cite[Chapter 20, page 125]{davenport} and \cite[Corollary 18.8]{IKbook}, respectively:

\begin{theorem}\label{th:prime_density}
	Given $a\in\ZZ_n^*$, write $\pi_{n,a}(x)$ for the number of primes less than $x$ that are congruent to $a$ modulo $n$. Then under GRH, there is an absolute constant $c>0$ such that
	\[ \pi_{n,a}(x) \geq \frac{x}{\varphi(n)\log x} - c x^{1/2}\log x.\]
	Unconditionally, there exist effective positive constants $c_1$ and $c_2$, such that for all $n < c_1 x^{c_1}$,
	\[ \pi_{n,a}(x) \geq \frac{c_2 x}{\varphi(n)x^{1/2}\log x}. \]
\end{theorem}

What remains to be understood is how to construct an arithmetic progression such that for every prime $p$ appearing in it, $\FF_p$ contains a representation $\overline{\alpha}_1,\ldots,\overline{\alpha}_k\in\FF_p$ of the square-root input $\sqrt{a_1},\ldots,\sqrt{a_k}$. As it turns out, this only
requires some classical results on quadratic reciprocity,
which we recall now.

Let $p$ be an odd prime number. An integer $a$ is said to be a \emph{quadratic residue} modulo $p$ if it is congruent to a perfect square modulo $p$, i.e., if there exists an integer $x$ such that $x^2 \equiv a \mod p$. The \emph{Legendre symbol} is a function of $a$ and $p$ taking values in~$\{1,-1,0\}$, that is
\[\left( \frac{a}{p} \right) =\begin{cases}
1 & \text{if $a$ is a quadratic residue mod $p$ and $a \not\equiv 0 \pmod{p}$},\\
-1 &\text{if $a$ is a non-quadratic residue mod $p$},\\
0 & \text{if $a\equiv 0 \pmod{p}$.} 	
\end{cases}
\]
  Its explicit definition is as follows:
\[ \left( \frac{a}{p} \right) =a^{\frac{p-1}{2}} \pmod{p}. \]
Furthermore, given  odd primes $p$ and $q$, the \emph{Law of quadratic reciprocity} states:
\[ \left( \frac{p}{q} \right)\left( \frac{q}{p} \right) = (-1)^{\frac{p-1}{2}\frac{q-1}{2}}.\]
We observe that in the case where $p\in 4\NN+1$, then
\[ \left( \frac{p}{q} \right)\left( \frac{q}{p} \right) = 1 \; \; \Leftrightarrow \; \; \left( \frac{p}{q} \right) = \left( \frac{q}{p} \right) = \pm 1.\]
In other words, $p$ is a quadratic residue modulo $q$ if and only if $q$ is a quadratic residue modulo~$p$, when either $p$ or $q \equiv 1 \pmod{4}$.

\begin{figure*}[t]
    \centering
    \begin{tabularx}{0.82\textwidth}{ p{0.1\textwidth} p{0.72\textwidth} }
   \rowcolor{Gray}
    \multicolumn{2}{c}{\textbf{Radical Identity Testing for square root inputs}} \\
    \textbf{Input:} & Algebraic circuit $C$ of size  at most $s$
    and a list of $k$ odd primes $a_1, \ldots, a_k$ 
    of magnitude at most $2^s$.
    \\
    \rowcolor{Gray}
    \textbf{Output:} & Whether $f(\sqrt{a_1},\ldots,\sqrt{a_k})=0$ for the polynomial $f(x_1,\ldots,x_k)$ computed by $C$.
    \\ 
    \textbf{Step 1:} &
            Compute $b$ such that $b +1 \equiv 5 \pmod{8}$, and $b +1 \equiv 1\pmod{a_i}$ for all $i$.
    \\  
    \textbf{Step 2:} & 
    		Pick $p$ uniformly at random from the set $S(a_1,\ldots,a_k)$ defined in~(\ref{eq:SA}).
     \\ 
    \textbf{Step 3:} &
            Compute $\overline{\alpha}_1,\ldots,\overline{\alpha}_k \in\FF_p$ such that $\overline{\alpha}_i^{2}  \equiv a_i \pmod{p}$ as described in
            \Cref{eq:pocklington}.
     \\  
    \textbf{Step 4:} &
            Output 'Zero' if $\overline{f}(\overline{\alpha}_1,\ldots,\overline{\alpha}_k) = 0$, where $\bar{f}$ is the reduction of $f$ modulo $p$; and 'Non-zero' otherwise.
    \\ \hline
    \end{tabularx}
    \caption{Our randomised polynomial time
   algorithm for the complement of  $2$-RIT.
    }
    \label{fig:corp_squareRIT}
\end{figure*}

In what follows, we apply the law of quadratic reciprocity in order to choose the right field $\FF_p$ for deciding $2$-RIT. Recall that intuitively, we are looking for a prime $p$ such that $x^2-a_i$ has a solution in $\FF_p$ for all $i$, that is, the $a_i$ is a quadratic residue modulo $p$. Since we will be choosing $p$ from an arithmetic progression, we can easily make that progression to be of the shape $4\NN + 1$, that is, to ensure that $p\equiv 1 \pmod{4}$. In that case the $a_i$'s will be quadratic residues modulo $p$ if and only if $p$ is a quadratic residue modulo $a_i$ for all~$i$. In order to ensure that, it suffices to choose $p$ such that $p \equiv 1 \pmod{a_i}$ for all $i$, as $1$ is a perfect square modulo $a_i$ for all $i$, and thus $p$ a quadratic residue modulo $a_i$.

We have just shown that if we choose $p$ such that it satisfies all the above-mentioned congruences, the polynomials $x^{2}-a_i$ all split and have non-zero roots in $\FF_p$. Following Pocklington's algorithm, there is a deterministic way to solve the equations $x^2-a_i$  if $p\equiv 5 \pmod{8}$. 
In particular, writing $p = 8m+5$, the solution of the equation $x^2 \equiv a \pmod{p}$ is given by the following function:
\begin{equation}\label{eq:pocklington}
x =\begin{cases}
\pm a^{m+1} & \text{if $a^{2m+1} \equiv 1 \pmod{p}$,}\\
\pm\frac{y}{2} &\text{if $a^{2m+1} \equiv -1 \pmod{p}$ and} \\
& \text{$y = \pm(4a)^{m+1}$ is even,} \\ 
\pm\frac{p+y}{2} &\text{if $a^{2m+1} \equiv -1 \pmod{p}$ and} \\
& \text{$y = \pm(4a)^{m+1}$ is odd.}
\end{cases}
\end{equation}
See \Cref{appendix:pocklington} for details.
Note that this congruence encompasses the above restriction on $p$ being congruent to 1 modulo~4.

We now show how to construct an arithmetic progression such that all primes $p$ in the progression satisfy the above congruences. Denote by $A=\prod_{i=1,a_i\neq 2}^{k}a_i$ the product of all input radicands $a_i$. Note that $A$ will always be odd as the $a_i$ are odd primes. Let us look at the arithmetic progression
\begin{equation}\label{eq:arithmetic_progression_corp}
	 8A\NN + b + 1,
\end{equation}
where $b$ is a solution of the following system of equations 
\begin{align}\label{eq:bAP}
	b \equiv 4 &\pmod{8} \\
	b \equiv 0 &\pmod{a_i}. \nonumber
\end{align}
Since all the moduli in the equations~(\ref{eq:bAP}) are pairwise coprime, by the Chinese remainder theorem, the  system has a solution.
By the construction above,
 we have an arithmetic progression such that all primes $p$ in the progression are  good primes.
We also ensure that $p$ is such that we can deterministically find the representations $\overline{\alpha}_1,\ldots,\overline{\alpha}_k$ of $\sqrt{a_1},\ldots,\sqrt{a_k}$ in $\FF_p$.
Define by $S(a_1,\ldots,a_k)$ the following set  
\begin{align}\label{eq:SA}\big\{p \leq 2^{5k^3} \mid  &p \in  8A \NN+ b+1  \text{ where } A=\prod_{i=1}^{k}a_i \\
    &\text{ and } b \text{ is a solution of~(\ref{eq:bAP})} \big\}. 	\nonumber
\end{align}
We have the following:

\begin{restatable}{proposition}{probabilityprimep}\label{prop:density_corp}
Let $C$ be an algebraic circuit of size at most $s$, and $a_1, \ldots, a_k$ odd primes of bit-length at most $s$. 
	 Denote by $\alpha$ the algebraic integer obtained by evaluating~$C$ on the~$\sqrt{a_i}$. Suppose that $p$ is chosen uniformly at random from the set~$S(a_1,\ldots,a_k)$ defined in~(\ref{eq:SA}).
	 Then
		\begin{itemize}
			\item[(i)] $p$ is prime with probability at least
			$\frac{1}{6s^3}$ assuming GRH, 
			 and
			\item[(ii)] given that $p$ is prime, the probability that it divides $N(\alpha)$ is at most $2^{-s^3}$	unconditionally.
		\end{itemize}
\end{restatable}

With~\Cref{prop:density_corp} in hand, we can state our algorithm, see \Cref{fig:corp_squareRIT}, and prove its complexity.

\begin{proof}[Proof of \Cref{th:squareRIT-coRP}]
\Cref{fig:corp_squareRIT} presents a 
randomised polynomial time algorithm for the complement of  $2$-RIT  as follows. It is clear that the algorithm runs in polynomial time. Let us now argue its correctness.
	
	First, suppose that $f(\sqrt{a_1},\ldots,\sqrt{a_k}) = 0$. By \Cref{th:RIT-soundess}, we have $\overline{f}(\overline{\alpha}_1,\ldots,\overline{\alpha}_k) = 0$, and hence the output is 'Zero'. Second, suppose that $f(\sqrt{a_1},\ldots,\sqrt{a_k}) \neq 0$. Then the output will be 'Non-Zero' provided that $p$ does not divide $N(f(\sqrt{a_1},\ldots,\sqrt{a_k}))$. By \Cref{prop:density_corp}(ii), the probability that $p$ does not divide $N(f(\sqrt{a_1},\ldots,\sqrt{a_k}))$ is at least $1 - 2^{-s^3}$. 
	Thus, the probability that the algorithm gives the wrong output is at most $2^{-s^3}$.

	It remains to show that  $2$-RIT  is in {\coNP} unconditionally. The idea is to modify the algorithm in \Cref{fig:corp_squareRIT}, replacing randomisation with guessing. \Cref{th:prime_density} shows that \begin{equation}\label{eq:numprimnocon}\pi_{8A,b+1}(2^{3s^3}) > 2^{s^3}\end{equation}
	 for $s$ sufficiently large.
	 It follows that there exists a prime $p\leq 2^{3s^3}$
	 that does not divide $N(f(\sqrt{a_1},\ldots,\sqrt{a_k}))$. The rest of the argument follows as in the proof of Theorem~\ref{th:rit-conp}.
	\end{proof}

Note that, in general, in the proof of the theorem above, GRH is required to obtain the {\coRP} bound. 
This is because the unconditional lower bound on density of primes in arithmetic progressions is not strong enough for our purposes: 
The number of primes less than $2^{3s^3}$ which are favourable is just $2^{s^3}$, as computed in \eqref{eq:numprimnocon}. This gives a probability of success at least $2^{-2s^3}$, which is exponentially small in the instance size. In order to get a constant success probability, we have to repeatedly sample and run this algorithm $2^{2s^3}$ times, which yields an exponential time algorithm. However, under GRH, the bound is improved to $\frac{1}{6s^3}$ and polynomially many repetitions suffice for a constant success probability.

\section{Discussion}
\label{sec:discussion}

In this work, we have shown that the RIT problem is in {\coNP} under GRH. We also showed that $2$-RIT is in {\coRP} assuming GRH, and is in {\coNP} unconditionally. Our algorithms  work by reducing the polynomials modulo a "small" prime~$p$, and taking the computation to the finite field $\FF_p$. We opt for this approach as the coefficients of polynomials represented by algebraic circuits could be doubly exponential in the size of the circuit, which presents problems for an approach via
numerical approximation. We note though the latter approach has been shown to work for deciding  Bounded-RIT  when the exponents~$d_i$ of the radical input~$\mysqrt{d_i}{a_i}$ are unary. In particular, \cite{blomer98} gives a randomised polynomial time algorithm that places  Bounded-RIT  with unary exponents in {\coRP}. 
If we reuse the same exact approach for  Bounded-RIT  with radical inputs with binary exponents, the algorithm no longer runs in polynomial time (the increase in complexity appears in Steps $1$ and $2$ of the algorithm in~\cite{blomer98}, when trying to compute the minimal $d_{ij}$ such that $\mysqrt{d_i}{m_j}^{d_{ij}}\in\ZZ$). Recall our reduction of RIT to its variant where the minimal polynomials of the input radicals $\mysqrt{d_i}{a_i}$ are $x^{d_i}-a_i$ and   all radicands~$a_i$ are pairwise coprime numbers, presented in~\Cref{appendix:reduction}. Applying this reduction first would 
avoid the above-mentioned increase in complexity for  Bounded-RIT  with binary exponents.
It thus places the general version of  Bounded-RIT  in 
{\coRP}.

\paragraph{\bf Working in extensions of finite fields.}
Our approach to RIT involves
finding a prime $p$ for which the required radicals 
exist in $\FF_p$.  Here it is tempting to consider working instead in a finite extension of $\FF_p$.  For example, it is well-known that every irreducible polynomial over $\FF_p$ of degree $d$ splits completely over $\FF_{p^d}$.
For solving RIT, a major problem with this approach is that if $d$ is given in binary then representing an element of the field $\FF_{p^d}$ requires space
exponential in the bit-length of $d$.  Specialising to 2-RIT, we have that the 
required square roots all exist in $\FF_{p^2}$, for any $p$.  But working in $\FF_{p^2}$ is only sound if the latter is a quotient of the number field $K$ generated by the square roots, that is, if $p$ has inertial degree~2 over $K$.  Moreover, the asymptotic density of such primes is the same as for those that split over $K$.

\SkipTocEntry\section*{Acknowledgements}
This work has been partially supported by ANR-CREST-JST CyphAI.

\bibliographystyle{plain}
\bibliography{literature}

\appendix

\section{Reduction to the RIT problem with coprime radicands}
\label{appendix:reduction}

Given a set of integers $a_1,\ldots,a_k$, the \emph{factor-refinement} algorithm~\cite{BACH1993199} computes a set~$\{m_1,\ldots,m_{\ell}\}$ of (not necessarily prime) factors~$m_j$ of the $a_i$'s such that 
 $gcd(m_j,m_k)=1$ for all $1\leq j<k\leq \ell$, and each $a_i$ can be written as a product of these factors, i.e.,  $a_i = \prod_{j=1}^{l} m_j^{e_{ij}}$ with the $e_{ij}\in\NN$. If we denote by $a = \lcm(a_1,\ldots,a_k)$, the factor-refinement algorithm runs in time $\mathcal{O}(\log^2(a))$(see also~\cite[Lemma 3.1]{blomer98}), and the number~$\ell$ of factors is bounded by $\sum_{i=1}^{k} \log(|a_i|)$. 
  
 \paragraph{Reduction.}
 Given an algebraic circuit $C$ representing a polynomial $f(x_1,\ldots,x_k)$ together with $k$ input radicals $\mysqrt{d_1}{a_1},\ldots,\mysqrt{d_k}{a_k}$, we construct another algebraic circuit~$C'$ representing an $\ell$-variate polynomial $f'(y_1,\ldots,y_{\ell})$ and input radicals $\mysqrt{t_1}{n_1},\ldots,\mysqrt{t_\ell}{n_\ell}$, with the $n_j$ pairwise coprime and respective minimal polynomials $x^{t_j}-n_j$, such that 
$f(\mysqrt{d_1}{a_1},\ldots,\mysqrt{d_k}{a_k})=0$ if and only if $f'(\mysqrt{t_1}{n_1},\ldots,\mysqrt{t_\ell}{n_\ell})=0$.

We first compute the partial factorisation of each one of the $a_i$'s by going through all primes up to $\log a$, which can clearly be done in $poly(\log a)$ time, where $a = \lcm(a_1,\ldots,a_k)$. We denote by $m_1,\ldots,m_r$, the primes $p$ appearing in the factorisations of the $a_i$. We then apply the \emph{factor-refinement} algorithm to the unfactored parts of the $a_i$'s 
and compute a set of pairwise coprime factors $\{m_{r+1},\ldots,m_\ell\}$ such that $a_i=\prod_{j=1}^{l} m_j^{e_{ij}}$. Then  
\begin{align*}
	f(\mysqrt{d_1}{a_1},\ldots,\mysqrt{d_k}{a_k})=0 \, & \Longleftrightarrow \, f\big(\prod_{j=1}^{l} m_j^{\frac{e_{1j}}{d_1}},\ldots,\prod_{j=1}^{l} m_j^{\frac{e_{kj}}{d_k}}\big)=0
\end{align*}

To construct the new input radicals $\mysqrt{t_i}{n_i}$ with respective minimal polynomials $x^{t_i}-n_i$, we compute for each $\mysqrt{d_i}{m_j}$ the smallest $d_{ij}$ such that $\mysqrt{d_i}{m_j}^{d_{ij}}\in\ZZ$. Observe that in general $m_j= p_1^{f_{j1}} \ldots p_s^{f_js}$ with $p_1,\ldots,p_s$ rational primes, and we have
\[\mysqrt{d_i}{m_j}^{d_{ij}} =
  \mysqrt{d_i}{p_1^{f_{j1}} \cdots p_s^{f_{js}}}^{d_{ij}} =
  \left( p_1^{f_{j1}} \cdots p_s^{f_{js}} \right) ^ {\frac{d_{ij}}{d_i}}
\]
which will be an integer if and only if $d_{i} | \gcd(f_{j1}, \ldots, f_{js}) \cdot d_{ij}$. Furthermore, observe that $d_{ij}$ will be the smallest such power precisely when
\begin{equation}\label{eq:di_gcd_dij}
    d_{i} = \gcd(f_{j1}, \ldots, f_{js}) \cdot d_{ij}.
\end{equation}

Now for the first $r$ factors of the $a_i$'s which are all prime, we have that $m_j = p$ for some rational prime $p$, and $d_{ij} = d_i$.

For $m_j$, $r < j \leq l$, first note that all $m_j$'s are products of powers of primes larger than $\log a$. Thus the multiplicities of the primes appearing in the decompositions $m_j= p_1^{f_{j1}} \ldots p_s^{f_js}$, that is, all $f_{j}$'s, are small. In particular, $f_{j} < \log m_j$ for all $j$, and furthermore $\gcd(f_{j1},\ldots,f_{js}) < \log m_j$.  Keeping this observation in mind, we can now show how to compute the $d_{ij}$ in time polynomial in $\log m_j$.
    
 Following \eqref{eq:di_gcd_dij}, we go trough the candidates for the $\gcd(f_{j1},\ldots,f_{js})$, $g_{ij} = 1,\ldots,\log m_j-1$, computing $f = \frac{d_i}{g_{ij}}$, the candidate for our $d_{ij}$ (note that if $f\notin\ZZ$, we discard it and move on). We then approximate $\mysqrt{d_i}{m_j}^f = m_j^{\frac{f}{d_i}} = m_j ^ {\frac{1}{g_{ij}}}$ with absolute error less than $1/2$ to obtain the unique integer $m$ with $|\mysqrt{g_{ij}}{m_j} - m| < 1/2$. This can be done by doing $\log m_j < \log a$ iterations of the Newton iteration \cite[Lemma~3.1]{Brent2016}. 
 We conclude by checking whether $m ^{d_i} = m_j^f$. This can be efficiently computed by writing $d_i = fg$ and observing that we are checking whether $(m^{g_{ij}})^f = m_j ^ e$, which simplifies to $m^{g_{ij} }=m_j$, with $g_{ij}$ unary.

Finally, we construct the new input radicals by setting $n_j = m_j ^{\frac{1}{\lcm(d_{1j}\cdots d_{kj})}}$ and $t_j = \frac{d_1\cdots d_k}{\lcm(d_{1j} \cdots d_{kj})}$.
We complete the reduction by constructing the algebraic circuit $C'$ from $C$ by replacing the leaves $x_i$, $i\in \{1,\ldots,k\}$, with a small circuit that computes $\prod_{j=1}^{l} y_j^{\frac{e_{ij}d_1\cdots d_k}{d_i}}$;
see Figure~\ref{fig:reduction}.

\begin{figure}[t]
\begin{center}
\begin{tikzpicture} 

\node[isosceles triangle, isosceles triangle apex angle=80, rotate=90,
    color=red!60, fill=red!5, draw, minimum size =2cm,thick] (T)at (0,0){}; 
 \node   at (0,0) (C) {\textcolor{red}{circuit $C$}};
       
 \filldraw[red] (-1.1,-.75) circle (3pt);
  \node at (-1.1, -1.1)   (a1) {$x_1$};
  \node at (-.25, -1.1)   (dots) {$\ldots$};
 \filldraw[red] (.5,-.75) circle (3pt);
  \node at (.5, -1.1)   (ak) {$x_k$}; 
 \filldraw[red] (1.2,-.75) circle (3pt);
  \node at (1.2, -1.1)   (const) {$-1$};

 \node   at (3,0) (CC) {$\Longrightarrow$};
 
 \node[isosceles triangle, isosceles triangle apex angle=80, rotate=90,
    color=blue!60, fill=blue!5, draw, minimum size =3.2cm,thick] (T)at (6,-.6){};  
       
  \node[isosceles triangle, isosceles triangle apex angle=80, rotate=90,
    color=red!60, dashed, draw, minimum size =2cm,thick] (T)at (6,0){}; 
 \node   at (6,0) (CC) {\textcolor{blue}{circuit $C'$}};

 \node[draw] at (5, -.7)   (t1) {$\times$};
 \node[draw] at (7, -.7)   (t2) {$\times$};

\node at (4, -2.1)   (y1) {$y_1$};
\node at (4.8, -2.1)   (yl) {$y_{2}$}; 
\node at (6.4, -2.1)   (a) {$\ldots$};

\node at (7.2, -2.1)   (yl) {$y_{\ell}$}; 
\node at (8, -2.1)   (m) {$-1$};

\filldraw[blue] (4,-1.75) circle (3pt);
 \path[->,draw,blue] (t1) edge  (4.1,-1.7);
 \filldraw[blue] (5,-1.75) circle (3pt);
  \path[->,draw,blue] (t1) edge  [bend left] (5.1,-1.7);
    \path[->,draw,blue] (t1) edge [bend right] (4.9,-1.7);
 \path[->,draw,blue] (t1) edge  (7.1,-1.7);
 \filldraw[blue] (7.2,-1.75) circle (3pt);
 \filldraw[blue] (8,-1.75) circle (3pt);
  
	\end{tikzpicture}
	\end{center}
	\caption{The scheme of the reduction from  RIT  to its variant where the input radicands are pairwise coprime and the exponents are all equal. In this simple example, $a_1$ is factored to $m_1m_2^2m_{\ell}$. }\label{fig:reduction}
\end{figure}
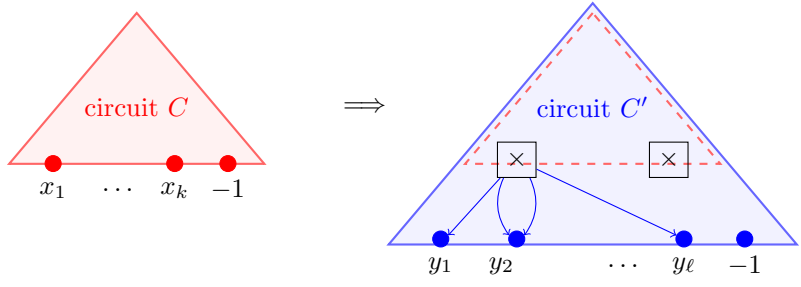

\section{Extended preliminaries}\label{appendix:preliminaries}
\subsection{Ring theory}

A \emph{ring} is a set $R$ equipped with two binary operations, addition ($+$) and multiplication ($\cdot$), satisfying the following three sets of axioms, called the ring axioms: $R$ is an abelian group under addition, $R$ is a monoid under multiplication, and multiplication is distributive with respect to addition. We further assume the addition and multiplication to be commutative and with unity. The most familiar example of a ring is the ring~$\ZZ$ of rational integers. 

Given a ring $R$, a subset $I$ of $R$ is said to be an \emph{ideal} if $I$ is an additive subgroup of the additive group of $R$ that absorbs multiplication by the elements of $R$. 
Given a rational prime~$p\in \ZZ$, the additive group $p\ZZ$ is an ideal of $\ZZ$.
Any ideal $I$ of $R$ that is not the whole of $R$ is said to be a \emph{proper ideal}, that is, the underlying set of $I$ is a proper subset of the underlying set of $R$. A proper ideal $I$ is called a \emph{prime ideal} if for any $a$ and $b$ in $R$, if $ab$ is in $I$, then at least one of $a$ and $b$ is in $I$.

An \emph{$R$-module} over a ring $R$ is a generalisation of the notion of vector space over a field, wherein scalars are elements of a given ring with identity and an operation of multiplication (on the left and/or on the right), called scalar multiplication,  defined between elements of the ring and elements of the module.

\subsection{Algebraic number theory}
A complex number $\alpha$ is \emph{algebraic} if it is a root of a univariate polynomial with integer coefficients. The defining polynomial of $\alpha$, denoted $f_\alpha$, is the unique (up to multiplication by $\pm 1$) integer polynomial of least degree, whose coefficients have no common factor, that has $\alpha$ as a root.
The \emph{degree} of an algebraic number $\alpha$ is the degree of its minimal polynomial $f_\alpha$.
If $f_\alpha$ is monic then we say that $\alpha$ is an \emph{algebraic integer}.
The sum, the difference, the product and the quotient of two algebraic numbers (except for division by zero) are algebraic numbers; this means that the set of all algebraic numbers is a \emph{field}, commonly denoted by $\Bar{\QQ}$. The sum, the difference, and the product of two algebraic integers is again an algebraic integer; given an algebraic field $K$, the algebraic integers of $K$, form a ring denoted $\OO_K$, called the ring of integers.

A field $K$ is said to be a \emph{field extension}, denoted $K/L$, of a field $L$, if $L$ is a subfield of $K$. Given a field extension $K/L$, the larger field $K$ is an $L$-vector space. The dimension of this vector space is called the \emph{degree} of the extension and is denoted by $[K : L]$.

An \emph{algebraic number field} (or simply number field) $K$ is a finite degree field extension of the field of rational numbers $\QQ$.
Thus $K$ is a field that contains $\QQ$ and has finite dimension when considered as a vector space over $\QQ$. 
It is well-known that each number field $K$ is a simple extension of $\QQ$, i.e., $K$ can be represented as $K = \QQ(\alpha)$, which is generated by the adjunction of a single element $\alpha \in K$, which is said to be the \emph{primitive element}.
 
The Gaussian rationals~$\QQ(i)$  are the first nontrivial example of an algebraic number field, obtained by adjoining $i:=\sqrt{-1}$ to~$\QQ$.  All elements of~$\QQ(i)$ can be written as expressions of the form $a+bi$ with $a,b\in \QQ$; hence $[\QQ(i) : \QQ]=2$. 
Furthermore, $\OO_{\QQ(i)}:=\ZZ[i]$.

An \emph{order}~$\OO$ in a number field $K$ is a free $\ZZ$-submodule of $\OO_K$ of rank $[K:\QQ]$. Since $\OO_K$ is also a free $\ZZ$-module of rank $[K:\QQ]$, it follows from the structure theorem for $\ZZ$-modules that the quotient $\OO_K/\OO$ is a finite abelian group. The order of this quotient, denoted $[\OO_K : \OO]$, is 
called the \emph{index} of $\OO$ in $\OO_K$.
It is known that $m\OO_K\subset \OO$ for $m=[\OO_K : \OO]$.
For example, $\ZZ[2i]=\ZZ+\ZZ2i$ is an order of the Gaussian integers of index $4$, and $4\ZZ[i] \subset \ZZ[2i]$.

Let $p(x) \in K[x]$ be a polynomial. The \emph{splitting field} of $p(x)$ over~$K$ is the smallest extension of~$K$ over which $p(x)$ can be decomposed into linear factors. The  splitting field of $x^2-1$ over~$\QQ$ is $\QQ(i)$.

A \emph{root of unity} is any complex number that yields 1 when raised to some positive integer power $n$, \emph{i.e.}, $\zeta$ such that $\zeta^n = 1$. If $\zeta_n$ is an $n$-th root of unity and for each $k < n$, $\zeta^k \neq 1$, then we call it a \emph{primitive $n$-th root of unity}. We can always choose a primitive $n$-th root of unity by setting $\zeta_n = e^{2i\pi\frac{k}{n}}$ for $k$ with $k\in \ZZ_n^*$.
The $n$-th \emph{cyclotomic polynomial}, for any positive integer~$n$, is the unique irreducible polynomial $\Phi(x) \in \QQ[x]$ with integer coefficients that is a divisor of~$x^n - 1$ and is not a divisor of $x^k - 1$ for any $k < n$. The $n$-th cyclotomic polynomial~$\Phi_n$ is the minimal polynomial of a primitive $n$-th root of unity, and its roots are all $n$-th primitive roots of unity.
 
 \subsection{Galois theory}

An algebraic field extension $K/L$ is \emph{normal} (in other words, $K$ is normal over $L$) if every irreducible polynomial over $L$ that has at least one root in $K$ splits completely over $K$. In other words, if $\alpha\in K$, then all conjugates of $\alpha$ over $L$ (i.e., all roots of the minimal polynomial of $\alpha$ over $L$) belong to $K$.  An algebraic field extension $K/L$ is said to be a \emph{separable extension} if for every $\alpha\in K$, the minimal polynomial of $\alpha$ over $L$ is a separable polynomial. That is, it has no repeated roots in any extension field. Every algebraic extension of a field of characteristic 0 is normal. A \emph{Galois extension} is an algebraic field extension that is normal and separable. The following holds for separable extensions:

\begin{theorem}[Primitive Element Theorem]
\label{th:primitive_element}
Let $K/L$ a separable extension of finite degree. Then $K = L(\alpha)$ for some $\alpha \in K$; that is, the extension is simple and $\alpha$ is a primitive element. 
\end{theorem}
 
Given a Galois extension $K/L$, the \emph{Galois group} of $K/L$, denoted by $Gal(K/L)$ is the group of automorphisms of $K$ that fix $L$. That is, the group of all isomorphisms $\sigma: K \to K$ such that $\sigma(x) = x $ for all $x \in L$.
If $K$ is a field with subfield $L\subset K$, the \emph{Galois closure} of $K$ over $L$ is the field generated by images of embeddings $K \to K$ that are the identity map on $L$.

Fix $\alpha$ an algebraic number $\alpha$ over a Galois extension $K/L$.
The image of $\alpha$ under an automorphism $\sigma\in \Gal(K/L)$ is called a \emph{Galois conjugate} of $\alpha$. The Galois conjugates of $\alpha$  are precisely the roots of the minimal polynomial $f_\alpha$ of $\alpha$. 
The Galois conjugates of a root of unity~$\zeta_n$ are its powers~$\zeta^k_n$ such that $k\in \ZZ^{*}_n$; and $Gal(\QQ(\zeta_n)/\QQ)$ includes all automorphisms $\sigma$ defined by $\sigma(\zeta_n) = \zeta_n^k$ for  $k\in \ZZ^*_n$.

The \emph{norm} of $\alpha$ is defined by
\[ N_{K/L}(\alpha) = \prod_{\sigma \in \Gal(K/L)} \sigma(\alpha)\]
For short, we may drop the subscript $K/L$ if the underlying field is understood from the
context. 
For $\alpha = a + bi \in\ZZ[i]$ the only Galois conjugate is $a - bi$, and thus  its norm is the product
$N(\alpha) = (a + bi)(a - bi) = a^2 + b^2$. Note that the norms of all Galois conjugate are equal, and the norm of an algebraic integer is always a rational integer itself.

The \emph{trace} of $\alpha \in K/L$ is defined by:
\[ \Tr_{K/L}(\alpha) = \sum_{\sigma \in \Gal(K/L)} \sigma(\alpha) \]
Again, we drop the subscript $K/L$ if the underlying field can be understood from the context.

The ring of integers of $K$, $\OO_K$, is a free abelian group of rank $n$, and hence admits $\ZZ$-basis $\{\alpha_1, \ldots, \alpha_n\}$. Given such a basis, we denote with $\Delta_K$ the \emph{discriminant}, and define it by
\[ \Delta_K = \det(\Tr_{K/L}(\alpha_i\alpha_j))_{1\leq i,j \leq n}\]
Note that $\Delta_K$ is always a non-zero rational integer.

\subsection{Ramification theory}
Let $K$ be a number field and let $\mf{p}$ be a prime ideal of $\OO_K$. Then $\mf{p} \cap \ZZ$ is a prime ideal of $\ZZ$, hence there must exist a rational prime  $p$ such that $\mf{p} \cap \ZZ = p\ZZ$. We say that $\mf{p}$ \emph{is above} $p$. We have:
\begin{align*} 
\mf{p} \ \subset \  & \OO_K \  \subset \ K \\
 &~\mid \\
 p \ \subset &\ \ZZ \ \subset \ \QQ. 
\end{align*}

Given a number field $K$ with ring of integers $\OO_K$, any ideal $I\subseteq\OO_K$ admits a unique factorisation into prime ideals in $\OO_K$. Let $p\in\ZZ$ be a rational prime. The ideal~$p\OO_K$ may not be prime in $\OO_K$, but does factorise into prime ideals as follows:
\begin{equation}\label{eq:ram} p\OO_K = \mf{p}_1^{e_1} \cdots \mf{p}_g^{e_g}. 
\end{equation}

Using the vocabulary introduced above, we can observe that the prime ideals $\mf{p}_i$ are all above $p$. 
Note that in the ring of integers of a number field, all prime ideals are maximal, hence all $\mf{p}_i$ are also maximal ideals of $\OO_K$.
In general, given a commutative ring $R$ and a maximal ideal $\mf{m}$ of $R$, the \emph{residue field} is the quotient $k = R/\mf{m}$. Intuitively, $k$ can be thought of as the field of possible remainders. Now, given a maximal ideal $\mf{p}$ of $\OO_K$, $\OO_K/\mf{p}$ is a $\FF_p$-vector space of finite dimension. The \emph{residual class degree} (\emph{inertial degree}), denoted $f_\mf{p}$, is the dimension of the $\FF_p$-vector space $\OO_K/\mf{p}$, that is:
\[ f_\mf{p} = \dim_{\FF_p}(\OO_K/\mf{p}).\]
Looking at the factorisation \eqref{eq:ram}, we can define the residue class degree of each one of the $\mf{p}_i$'s as follows:
    \[ f_i = [ \OO_K/\mf{p}_i : \ZZ/p\ZZ ]. \]

We say that $p$ is \emph{ramified} if the \emph{ramification index} $e_i > 1$ for some $\mf{p}_i$. A prime $p$ is said to be \emph{totally ramified} if $e=n$, $g=1$, and $f=1$. That is $p\OO_K = \mf{p}^{e}$ for some $\mf{p}$.
 Conversely, $p$ is \emph{non-ramified} if $p\OO_K = \mf{p}_1 \cdots \mf{p}_g$ where the $\mf{p}_i$ are distinct. We further say that a prime $p\in\ZZ$ is \emph{inert} if the ideal $p\OO_K$ is prime, in which case we have $p\OO_K = \mf{p}$, that is $g=1$, $e=1$, and $f=n$.
 Finally, a prime is said to be \emph{split} if $e_1 = \ldots = e_g = 1$.
 
For the Gaussian integers, the ideals $2\ZZ[i]$ and $5\ZZ[i]$ are not  prime ideals and have respective factorisations $2\ZZ[i]=\mf{p}^2$ and $5\ZZ[i]=\mf{p}_1\mf{p}_2$ where $\mf{p}=(1+i)\ZZ[i]$, $\mf{p}_1=(2+i)\ZZ[i]$, and $\mf{p}_2=(2-i)\ZZ[i]$ are prime ideals. The prime $2$ is the unique ramified prime in the Gaussian integers.

\subsection{The $p$-adic field $\QQ_p$}

Here, we give a brief preliminary on the field of $p$-adic numbers $\QQ_p$; for more details see, e.g., \cite{gouvea2003p}. 
The field $\QQ_p$ is the completion of $\mathbb{Q}$ with respect to the $p$-adic absolute value $|\cdot|_p$, given by 
$|a/b|_p = p^{v_p(b)-v_p(a)}$ for $a,b\in \mathbb{Z}\setminus \{0\}$,
where $v_p(x)$ denotes the order to which $p$ divides $x \in \mathbb{Z}$. We denote by $\mathbb{Z}_p$ the valuation ring 
$\{ x \in \mathbb{Q}_p : |x|_p \leq 1\}$.
This is the local ring with unique maximal ideal generated by $p$.
A basic result about $\mathbb{Q}_p$ is Hensel's Lemma:

\begin{lemma}
\label{lem:hensel}[Hensel's lemma]
	Given $f(X) \in \mathbb{Z}[X]$, if there exists $\alpha \in \mathbb{F}_p$ such that 
	\[f(\alpha)=0 \text{ and } f'(\alpha)\neq 0\]
then there exists $x \in \mathbb{Z}_p$ with $f(x)=0$ and $x\equiv \alpha \bmod p$.
\end{lemma}

Given a number field $K/\QQ$, let $p$ be a rational prime and $\mathfrak{p}$ a prime ideal of $\OO_K$ lying above $p$. Then the $p$-adic absolute value $|\cdot|_p$ corresponding to $p$ extends uniquely to an absolute value $|\cdot|_\mathfrak{p}$ corresponding to $\mathfrak{p}$ such that the restriction of $|\cdot|_\mathfrak{p}$  to $\QQ$ coincides with $|\cdot|_p$.
This, in turn, corresponds to a field extension $K_{\mathfrak{p}}/\QQ_p$. This extension can be analysed using the ramification of $p$ in $K$. In particular, if $p$ completely splits in $K$, then $[K_{\mathfrak{p}} : \QQ_p] = 1$, that is, the extension is trivial and we have $K_{\mathfrak{p}} = \QQ_p$. If $p$ is inert, then the degree of the extension $K_{\mathfrak{p}}$ over $\QQ_p$ is equal to the inertial degree of $p$ in $K$. Finally, if $p$ is totally ramified, then the degree of the extension $K_{\mathfrak{p}}$ over $\QQ_p$ is equal to the degree of $K$ over $\QQ$, i.e., $[K_{\mathfrak{p}} : \QQ_p] = [K : \QQ]$.

Given a prime ideal $\mathfrak{p}$ of $\OO_K$, we define the decomposition group $D_{\mathfrak{p}}$ of to be the set of all automorphisms $\Gal(K/\QQ)$ fixing $\mathfrak{p}$, that is $ D_{\mathfrak{p}} = \{ \sigma \in \Gal(K/\QQ) \mid \sigma(\mathfrak{p}) = \mathfrak{p} \}$.
If the field $K$ is Galois over $\QQ$, the following isomorphism holds:
\[ D_{\mathfrak{p}} \cong \Gal(K_{\mathfrak{p}}/\QQ_p). \]
This entails that the prime $p$ completely splits in $K$ if and only if the decomposition group $D_{\mathfrak{p}}$ is trivial for all prime factors $\mathfrak{p}$ of $p$ in $\OO_K$.

\section{Missing proofs from Section~\ref{sec:conp}}
\label{appendix:conp}

\subsection{Bound on the norm}\label{appendix:norm}
In this section, we prove a bound on the norm of the algebraic integer computed by an algebraic circuit $C$ on a radical input $\mysqrt{d_1}{a_1},\ldots,\mysqrt{d_k}{a_k}$, with the $a_i$ pairwise coprime, and the $d_i$ and $a_i$ of magnitude at most $2^s$, that is an instance of the RIT problem.
Let $d=\lcm(d_1,\ldots,d_k)$, and denote by $\zeta_d$ a primitive $d$-th root of unity. Let  $K = \QQ(\mysqrt{d_1}{a_1},\ldots,\mysqrt{d_k}{a_k}, \zeta_d)$. We claim:

\boundnorm*

\begin{proof}
Recall that $s$ is an upper bound on the number $k$ of input radicands, the size of the circuit, and that the magnitude of $a_i$ and $d_i$ is at most $2^s$. 

Write $\alpha = \sum_{i} b_i x_1^{e_{i_1}} \cdots x_k^{e_{i_k}}$ where
$e_{i_1} + \ldots + e_{i_k} \leq 2^s $, $b_i \in \ZZ$ with $b_i\leq 2^{2^s}$,
$i$ ranges over all monomials of the shape $x_i^{e_{i_1}}\cdots x_i^{e_{i_k}}$. Let us denote by $M$ the number of all such monomials, and count how many of them we can construct. Denote by $D = \max (d_1,\ldots,d_k)$, then:
\[ M = \binom{k+D}{D} = \binom{k + D}{k} \leq  \binom{s + 2^s}{s}
\leq (s + 2^s)^s \leq 2^{s^2}  \]

Denote by $G = \Gal(K/\QQ)$, and
note that given the size bounds on our input $|G| \leq 2^{2s^2}$.

 Observe that the action of all $\sigma\in G$ is determined by their action on $\zeta_d$, that is:
\[ \sigma(\mysqrt{d_i}{a_i}) = \mysqrt{d_i}{a_i}\sigma(\zeta_d).\]
Then
\begin{align*}
    N(\alpha)
    &= N\left(\sum_{i = 1}^M b_i x_1^{e_{i_1}} \cdots x_k^{e_{i_k}} \right) \\
    &= \prod_{\sigma\in G} \sigma\left(\sum_{i = 1}^M b_i x_1^{e_{i_1}} \cdots x_k^{e_{i_k}} \right) \\
    &= \prod_{\sigma\in G} \sum_{i=1}^M b_i \left(\max_{j=1}^k x_j \right)^{\#e_i}\sigma \left(\zeta_d \right)^{\#e_i}
    \text{ where } \#e_i = e_{i_1} + \ldots e_{i_k} = \sum_{j=1}^k e_{i_j} \\
    &= \prod_{l\in \ZZ_{|G|^*}} \sum_{i=1}^M b_i \left(\max_{j=1}^k x_j\right)^{\#e_i} \left(\zeta_d^l\right)^{\#e_i}
\end{align*}

Finally, putting together all of our bounds
yields
\begin{align*}
    |N(\alpha)|
    &\leq \prod_{l=1}^{2^{2s^2} } \left( \sum_{l=1}^{2^{s^2}} 2^{2^s} \cdot (2^s)^{2^s} \right) \\
    &\leq \prod_{l=1}^{2^{2s^2} } \left( \sum_{l=1}^{2^{s^2}} 2^{2^s (s+1)}\right) \\
    &\leq \prod_{l=1}^{2^{2s^2} } \left( 2^{s^2} \cdot 2^{2^s (s+1)} \right) \\
    &\leq \left( 2^{s^2} \cdot 2^{2^s (s+1)} \right) ^ {2^{2s^2}} \\
    &\leq 2^{2^{2s^2} (s2^s + 2^s + s^3)} \\
    &\leq 2^{2^{s^3}} \text{ for } s \geq 4.
\end{align*}
\end{proof}

\subsection{On the primitive element of a radical field extension}\label{appendix:primitive_element}
The aim of this section is to prove the bounds on the degree of the primitive element $\theta$ of our number field, and the size of the constants used in the linear combination of the field generators that we use to construct it.

Let us first recall the setting for our number field. Given $k$ radicals $\mysqrt{d_1}{a_1},\ldots,\mysqrt{d_k}{a_k}$, with the $a_i$ pairwise coprime, and the $d_i$ and $a_i$ of magnitude at most $2^s$. Let $d=\lcm(d_1,\ldots,d_k)$, and denote by $\zeta_d$ a primitive $d$-th root of unity. We would like to construct the primitive element $\theta$ for the number field $K = \QQ(\mysqrt{d_1}{a_1},\ldots,\mysqrt{d_k}{a_k}, \zeta_d)$. Note that besides adjoining the radicals to $\QQ$, we also make sure to add $\zeta_d$, which ensures our field extension $K$ is Galois.

The primitive element for number fields is as follows:

\begin{theorem}\label{th:primitive_el_theorem_number_fields}

Let $K = L(\alpha_1,\ldots,\alpha_k)$ be a finite extension of $L$, and assume that $\alpha_2,\ldots,\alpha_k$ are separable over $L$. Then there is an element $\theta\in K$ such that $K=L(\theta)$.
\end{theorem}

The proof of the above theorem is constructive (see, e.g., \cite[Theorem~4.1.8]{cohen2013course} or \cite[Theorem 5.1]{milneGalTheory}), and computes the primitive element $\theta$ as a linear combination of the generators $\alpha_1,\ldots,\alpha_k$, that is $\theta = \sum_{i=1}^k c_i \alpha_i$. The computation of $\theta$ is to be done inductively, constructing first a primitive element $\theta_2$ for $L(\alpha_1,\alpha_2)$, then $\theta_3$ for $L(\alpha_1, \alpha_2,\alpha_3)$, and so on until $\theta_k$. Furthermore, it is shown that only finitely many combinations of the constants $c_i$ fail to generate a primitive element for the field extension $K$. This gives rise to an effective version of the primitive element theorem for number fields that induces a bound on the degree of the algebraic number $\theta$, as well as on the size of the constants $c_i$. In particular, for a field generated by two algebraic numbers, the bounds are as follow (see \cite[Proposition 6.6]{kururPhd}):

\begin{proposition}\label{prop:effective_primitive_el_th}
    Let $\alpha$ and $\beta$ be algebraic numbers of degree $m$ and $n$ respectively. There exists an integer $c\in\{1,\ldots,m^2n^2+1\}$ such that $\alpha+c\beta$ is a primitive element of $\QQ(\alpha,\beta)$.
\end{proposition}

Note that in general, we could choose $c_i\in\QQ$ with only finitely many combinations of the $c_i$'s not giving us a primitive element. Thus the primitive element need not be an algebraic integer in general. However, we make the choice to choose $c_i\in\ZZ$, therefore the minimal polynomial $f_\theta$ of $\theta$ is monic and $\theta$ an algebraic integer.

Now let us prove our claim on the bounds on our primitive element $\theta$:
\boundprimitiveelement*

\begin{proof}
Note that given $d_1,\ldots,d_k \leq 2^s$, their least common multiple is at most of size $2^{s^2}$, hence we have:
\[ \deg \zeta_d \leq  2^{s^2}.  \]

In the spirit of the proof of the primitive element theorem, use \Cref{prop:effective_primitive_el_th} inductively as follows:
\begin{align*}
    \theta_2 = \mysqrt{d_1}{a_1} + c_2\mysqrt{d_2}{a_2} &\qquad c_2 \leq (2^s)^2(2^s)^2 + 1 \textcolor{blue}{\leq 2^{4s}} \\
    \textcolor{blue}{\deg \theta_2 \leq 2^{2s}} \\
    \theta_3 = \mysqrt{d_1}{a_1} + c_2\mysqrt{d_2}{a_2} + c_3\mysqrt{d_3}{a_3}  &\qquad c_3 \leq (2^{2s})^2(2^s)^2 + 1 \textcolor{blue}{\leq 2^{6s}} \\
    \textcolor{blue}{\deg \theta_3 \leq 2^{3s}} \\
    \vdots & \\
    \theta_k = \mysqrt{d_1}{a_1} + c_2\mysqrt{d_2}{a_2} + \ldots + c_k\mysqrt{d_k}{a_k} &\qquad c_k \leq (2^{(k-1)s})^2(2^s)^2 + 1 \textcolor{blue}{\leq 2^{2s^2}}\\
    \textcolor{blue}{\deg \theta_k \leq 2^{s^2}} \\
    \theta = \theta_k + c_0 \zeta_d &\qquad c_0 \leq (2^{ks})^2(2^{s^2})^2 + 1 \textcolor{blue}{\leq 2^{4s^2}}\\
    \textcolor{blue}{\deg \theta_2 \leq 2^{2s^2}} \\
\end{align*}

The claimed bounds on the degree of $\theta$ and the size of the constants $c_i$ follow.
\end{proof}

\subsection{Bound on the discriminant}\label{appendix:discriminant}
The aim of this section is to prove a bound on the discriminant of the minimal polynomial of the primitive element of our number field $K$. Recall we denote by $K = \QQ(\mysqrt{d_1}{a_1},\ldots,\mysqrt{d_k}{a_k}, \zeta_d)$, where $d=\lcm(d_1,\ldots,d_k)$, and $\zeta_d$ is a primitive $d$-th root of unity, and compute $\theta$ as a linear combination of $\mysqrt{d_1}{a_1},\ldots,\mysqrt{d_k}{a_k}, \zeta_d$. Note also that we assume the magnitude of the $d_i$'s and $a_i$'s to be at most $2^s$.

Recall that given a polynomial $f(x) = a_nx^n + \ldots + a_1x + a_0$ with roots $r_1,\ldots,r_n$, its discriminant can be computed as
\begin{equation}\label{eq:disc_formula}
    \Delta_{f} = a_n^{2n-2} \prod_{i<j} (r_i-r_j)^2 = (-1)^{\frac{n(n-1)}{2}} a_n^{2n-2} \prod_{i\neq j}(r_i-r_j)
\end{equation}

Now let us see how we use it to prove our claim:
\bounddiscriminant*

\begin{proof}
Denote by $G = \Gal(K/\QQ)$ the Galois group of $K$. Recall we construct the primitive element $\theta$ of our radical number field as a linear combination of the radicals $\mysqrt{d_1}{a_1},\ldots,\mysqrt{d_k}{a_k}$ and a primitive $d$-th root of unity $\zeta_d$ as follows:
\[ \theta = c_o \zeta_d + \sum_{i=1}^k c_i \mysqrt{d_i}{a_i}. \]
We choose the $c_i\in\ZZ$, hence $\theta$ is an algebraic integer. The minimal polynomial $f_\theta$ of the primitive element $\theta$ has roots $\theta = \theta_1, \ldots, \theta_{|G|}$. Note that the roots of $f_\theta$ are given by the elements of $G$, that is, $\theta_i = \sigma_i(\theta)$ for some $\sigma_i \in G$. Recall also that the elements of the Galois group $G$ act on conjugates of a given element of $K$ by permuting the $d$th roots of unity, that is, given $\alpha \in K$, $\sigma_i(\alpha) = \alpha \zeta_d^i$ for some and $\sigma_i \in G$.

Note that given $d_1,\ldots,d_k \leq 2^s$,
their least common multiple is at most of size $2^{s^2}$,
hence we have:
\[ \deg \zeta_d \leq  2^{s^2}   \]
As stated in \Cref{lem:primitive_el}, the constants $c_i$ in the computation of the primitive element can be bound by:
\[ c_i \leq 2^{4s^2}. \]

Now given $\sigma_j\in G$, write
\begin{align*}
    \sigma_j(\theta) &= \sigma_j\left(c_0 \zeta_d + \sum_{i=1}^k c_i \mysqrt{d_i}{a_i}\right) \\
                     &= c_0 \sigma_j(\zeta_d) + \sum_{i=1}^k c_i \sigma_j(\mysqrt{d_i}{a_i})  \\
                     &= c_0 \zeta_d^{j+1} + \sum_{i=1}^k c_i \mysqrt{d_i}{a_i} \zeta_d^{ij}
\end{align*}

Then for $1\leq j, l\leq |G|$, $j \neq l$
\begin{align*}
    \sigma_j(\theta) - \sigma_l(\theta)
    &= \left(c_0 \zeta_d^{j+1} + \sum_{i=1}^k c_i \mysqrt{d_i}{a_i} \zeta_d^{ij}\right) - \left(c_0 \zeta_d^{l+1} + \sum_{i=1}^k c_i \mysqrt{d_i}{a_i} \zeta_d^{il}\right) \\
    &= c_0\left(\zeta_d^{j+1} - \zeta_d^{l+1}\right) + \left(\sum_{i=1}^k c_i \mysqrt{d_i}{a_i}\right)\left(\zeta_d^{ij} - \zeta_d^{il}\right)
\end{align*}
Note that for any two $d$th roots of unity $\zeta_d^j, \zeta_d^l$, we always have $\zeta_d^j - \zeta_d^l \leq 2$, hence
\begin{align*}
    \sigma_j(\theta) - \sigma_l(\theta)
    &\leq 2c_0 + 2\left(\sum_{i=1}^k c_i \mysqrt{d_i}{a_i}\right) \\
    &\leq 2\cdot 2^{4s^2} + 2 \left(\sum_{i=1}^s 2^{4s^2} \cdot 2^s\right) \\
    &\leq 2^{4s^2 + 1} + 2^{4s^2+s+1}s \\
    &\leq 2^{2s ^ 3} \text{ for } s\geq 4.
\end{align*}
Thus
\begin{align*}
    |\Delta_{f_\theta}| &= |\prod_{j\neq l} (\sigma_j(\theta) - \sigma_l(\theta))| \\
    &\leq \left(2^{2s^3}\right)^{|G|^2} \\
    &\leq \left(2^{2s^3}\right)^{\left(2^{2s^2}\right)^2} \\
    &\leq \left(2^{2s^3}\right)^{2^{4s^2}} \\
    &\leq 2^{2s^3 \cdot 2^{4s^2}} \\
    &\leq 2^{2^{5s^2}} \text{ for } s\geq 4.
\end{align*}
\end{proof}

\subsection{Bound on the number of split primes}\label{appendix:bound_number_of_primes}
The aim of this section is to prove the following proposition:

\boundnprimes*

\begin{proof}
Recall the bound on $\pi_{1}$ given in \Cref{prop:number_of_primes}:
\[ \pi_{1}(x) \geq \frac{1}{|\Gal(K/\QQ)|} \left[ \pi(x) - \log \Delta_{K} - c x^{1/2} \log (\Delta_{K} x^{|\Gal(K/\QQ)|}) \right] \]

Recall also that $\pi(x)$, that is, the number of primes $\leq x$, can be bound as
\[ \pi(x) \geq \frac{x}{\log x} \]

Finally, 
note that $|\Gal(K/\QQ)| \leq 2^{2s^2}$.

Now compute
\begin{align*}
	\pi_{1}(2^{4s^3})
	&\geq \frac{2^{4s^3}}{2^{2s^2}\cdot 4s^3} - 2^{3s^2} - c\cdot 2^{2s^3} \cdot 2^{3 s^2} - c \cdot 2^{2s^3} \cdot 4 s^3\\
	&\geq \frac{2^{4s^3}}{2^{s^3}} - 2^{2s^3} - c\cdot 2^{2s^3} \cdot 2^{3 s^2} - c \cdot 2^{2s^3} \cdot 4 s^3 \quad \quad   \text{ for } s\geq 4 \\
	&\geq 2^{3s^3} - 2^{2s^3}(1 + c\cdot 2^{3s^2} + c \cdot 4s^3) \\
	&\geq 2^{2s^3}(2^{s^3} - c\cdot 2^{3s^2} - c \cdot 4s^3)\\
	&\geq 2^{2s^3} \text{ for a fixed constant } c \text { and } s \geq \max(c,5) \\
	&\geq 2^{s^3} + 1 .
\end{align*}

\end{proof}

\section{Missing proofs from Section~\ref{sec:corp}}
\label{appendix:corp}

\subsection{A note on Pocklington's algorithm}\label{appendix:pocklington}
Pocklington's algorithm is a technique for solving congruences of the form $x^2 \equiv a \mod p$, where $x$ and $a$ are integers and $a$ is a quadratic residue modulo $p$. Given an integer $a$ and odd prime $p$ as input, the algorithm separates three cases for $p$, and then computes $x$ accordingly. We are interested in the case where  $p = 8m + 5$ for some $m\in\NN$, as we note that $x$ can be computed deterministically as follows:

\begin{figure}[H]
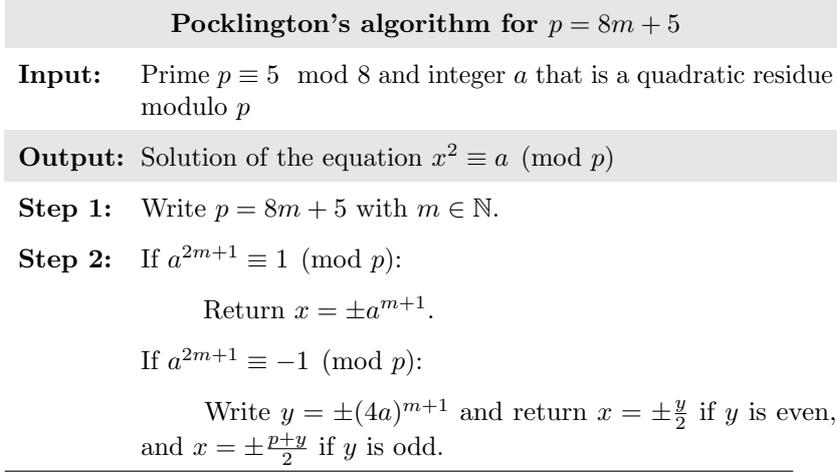

    \centering
    \begin{tabularx}{0.82\textwidth}{ p{0.1\textwidth} p{0.72\textwidth} }
   \rowcolor{Gray}
    \multicolumn{2}{c}{\textbf{Pocklington's algorithm for $p=8m+5$}}
    \\
    \textbf{Input:} & Prime $p \equiv 5 \mod 8$ and integer $a$ that is a quadratic residue modulo $p$ \\ 
    \rowcolor{Gray}
    \textbf{Output:} & Solution of the equation $x^2 \equiv a \pmod{p}$
    \\ 
    \textbf{Step 1:} &
            Write $p = 8m + 5$ with $m\in\NN$.
  \\
    \textbf{Step 2:} &
       If $a^{2m+1} \equiv 1 \pmod{p}$:
    \\ & $\qquad$ Return $x = \pm a^{m+1}$.
    \\ & If $a^{2m+1} \equiv -1 \pmod{p}$:
    \\ & $\qquad$ Write $y = \pm(4a)^{m+1}$ and return $x=\pm\frac{y}{2}$ if $y$ is even, and
    $x=\pm\frac{p+y}{2}$ if $y$ is odd.\\
   \hline
       \end{tabularx}
    \caption{Procedure to solve the congruence $x^2\equiv a \mod p$ for prime $p \equiv 5 \mod 8$.}
    \label{fig:pocklington}
\end{figure}

Let us briefly elaborate on the correctness of the above procedure. Given $p = 8m + 5$ with $m\in\NN$, following Fermat's little theorem, we have $x^{8m+4} \equiv 1 \mod p$. Since $x^2\equiv a \mod p$, we can rewrite it as $a^{4m+2} \equiv 1 \mod p$. Now, let us look at two separate cases for $a$:
\begin{enumerate}
    \item $a^{2m+1} \equiv 1 \mod p$ \\
    Then $a^{2m+2} \equiv a \mod p$, that is $(a^{m+1})^2 \equiv a \mod p$, hence $x = \pm a^{m+1}$.
    \item $a^{2m+1} \equiv -1 \mod p$ \\
    Note that 2 is a quadratic non-residue, so $4^{2m+1}\equiv -1 \mod p$, hence $4^{2m+1}a^{2m+1}\equiv 1 \mod p$. That is $(4a)^{2m+1}\equiv 1 \mod p$ and following the reasoning from the previous case $y = \pm(4a)^{m+1}$ is a solution of $y^2=4a$, hence $x=\pm\frac{y}{2}$, or if $y$ is odd, $x=\pm\frac{p+y}{2}$.
\end{enumerate}

\subsection{The probability of choosing a good prime $p$}\label{appendix:probability_good_prime}
The aim of this section is to prove a bound on the probability of randomly choosing a good prime $p$ in the randomised polynomial time algorithm for the complement of $2$-RIT.

\probabilityprimep*

\begin{proof}

We follow the proof of \cite[Proposition 9]{balaji2021cyclotomic}. Recall that we set  $a_i\leq 2^s$,
which implies $A \leq 2^{s^2}$.

For (i), we note that by \Cref{th:prime_density}, the probability that $p$ is prime is at most
\begin{align*}
	\frac{\pi_{8A,b+1}(2^{5s^3})}{{2^{5s^3}}/{8A}}
	&\geq \frac{8A}{\varphi(8A)\log 2^{5s^3} } - \frac{c \log 2^{5s^3} 8A}{(2^{5s^3})^{1/2}} \\
	&\geq \frac{1}{5s^3} - \frac{c {5s^3} 2^{s^2 + 3}}{2^{2s^3}} \\
	&\geq \frac{1}{5s^3} - \frac{c {5s^3} 2^{s^3}}{2^{2s^3}}
\end{align*}
where $c$ is the absolute constant mentioned in the theorem. For $k$ sufficiently large, the above is
$\frac{1}{6s^3}$, which proves the claim.
	
For (ii), by \Cref{lem:norm} the norm of $\alpha$ has absolute value at most
$2^{2^{s^3}}$, and hence $N(\alpha)$ has at most $2^{s^3}$ 
distinct prime factors. Then, for $s$ sufficiently large, the probability that $p$ divides $N(\alpha)$ given that $p$ is prime is at most
\[ \frac{6s^3 \cdot 8A \cdot 2^{s^3}}{2^{5s^3}} \leq 2^{-s^{3}}\]	
\end{proof}

\end{document}